\documentclass[12pt, draftcls, onecolumn, twoside]{IEEEtran}
\usepackage{graphicx} 

\usepackage[utf8]{inputenc} 
\usepackage[T1]{fontenc}
\usepackage{url}
\usepackage{ifthen}
\usepackage[noadjust]{cite}
\usepackage[cmex10]{amsmath} 
\usepackage{amssymb}
\usepackage{graphicx}
\usepackage{epstopdf}
\usepackage{epsfig}
\usepackage{multirow}
\usepackage{amsthm}
\usepackage{comment}
\usepackage{caption}
\usepackage{enumitem}
\usepackage{amsmath}
\usepackage{color}
\usepackage{bbm}
\usepackage{algorithm}
\usepackage{algpseudocode}
\usepackage{csquotes}
\usepackage{subfiles}
\usepackage{caption, multirow, makecell}
\usepackage{array}
\usepackage{caption}
\usepackage{nicematrix}
\usepackage{cuted}
\usepackage{subcaption}
\usepackage{mathtools}
\usepackage{blkarray}
\DeclarePairedDelimiter\ceil{\lceil}{\rceil}
\DeclarePairedDelimiter\floor{\lfloor}{\rfloor}
\usepackage[colorlinks=true, linkcolor=black, citecolor=black, urlcolor=black]{hyperref}

\newtheorem{thm}{Theorem}
\newtheorem{lem}{Lemma}

\newtheorem{cor}{Corollary}
\newtheorem{example}{Example}

\newtheorem{defn}{Definition}

\newtheorem{rem}{Remark}

\allowdisplaybreaks

\def\BibTeX{{\rm B\kern-.05em{\sc i\kern-.025em b}\kern-.08em
    T\kern-.1667em\lower.7ex\hbox{E}\kern-.125emX}}

\title{Fundamental Limits of Distributed Computing for Linearly Separable Functions\thanks{This research was partially supported by European Research Council ERC-StG Project SENSIBILITÉ under Grant 101077361, the ERC-PoC Project LIGHT under Grant 101101031, by the Huawei France-Funded Chair Toward Future Wireless Networks, and by the Program “PEPR Networks of the Future” of France 2030.}}
\author{K. K. Krishnan Namboodiri, Elizabath Peter, Derya Malak, and Petros Elia\thanks{The authors are with the Communication Systems Department, EURECOM, Sophia Antipolis, FRANCE (\{karakkad,peter,malak,elia\}@eurecom.fr).\hfill Manuscript last revised: {\today}.}}

\begin{document}

\maketitle
\begin{abstract}
This work addresses the problem of distributed computation of linearly separable functions, where a master node with access to $K$ datasets, employs $N$ servers to compute $L$ user-requested functions, each defined over the datasets. Servers are instructed to compute subfunctions of the datasets and must communicate computed outputs to the user, who reconstructs the requested outputs. The central challenge is to reduce the per-server computational load and the communication cost from servers to the user, while ensuring recovery for any possible set of $L$ demanded functions.

We here establish the fundamental communication–computation tradeoffs for arbitrary $K$ and $L$, through novel task-assignment and communication strategies that, under the linear-encoding and no-subpacketization assumptions, are proven to be either exactly optimal or within a factor of three from the optimum. In contrast to prior approaches that relied on fixed assignments of tasks -- either disjoint or cyclic assignments -- our key innovation is a nullspace-based design that jointly governs task assignment and server transmissions, ensuring exact decodability for all demands, and attaining \emph{optimality over all assignment and delivery methods}. To prove this optimality, we here uncover a duality between nullspaces and sparse matrix factorizations, enabling us to recast the distributed computing problem as an equivalent factorization task and derive a sharp information-theoretic converse bound. Building on this, we establish an additional converse that, for the first time, links the communication cost to the covering number from the theory of \emph{general covering designs}. 
\end{abstract}
\begin{IEEEkeywords}
Coded distributed computation, linearly separable functions, covering designs, matrix factorization, communication-computation tradeoff. 
\end{IEEEkeywords}
\section{Introduction}

Recent advances in machine learning have rendered distributed computing indispensable for managing increasingly intensive computing workloads. 
This distributed paradigm -- as exemplified by frameworks such as MapReduce \cite{DeG} and Spark \cite{ZCFSS} -- asks for decomposition of large-scale computational tasks into smaller subtasks and their execution in parallel across multiple servers, thereby overcoming the limited capacity of individual nodes to process massive datasets.

Despite their success, distributed computing systems face several fundamental challenges such as having to contend with unreliable or slow computing nodes (so-called stragglers \cite{LMA, OUG,LLPPR,DFHJCG,YMA,RDT}), protecting the privacy of sensitive datasets, and guarding against adversarial behavior by malicious nodes \cite{XBW, WSJC3, YuS, YSun}. Beyond these issues of reliability and security, two intrinsic constraints dominate system design: servers have finite computational power, and communication links have finite capacity. These intertwined constraints bring to the fore the well-known \emph{communication–computation tradeoff}, which lies at the heart of distributed computing across diverse settings such as MapReduce \cite{LMYA, YaYW, PLE, BrEl,PNR,WaW, WCJ}, gradient coding \cite{YAb, RPPA, TLDK}, and distributed matrix multiplication \cite{YMA, RTV, LSR, DFHJCG}, to mention just a few. These same constraints have inspired the development of various coding-theoretic techniques, proposed to reduce communication load, indeed often at the expense of increased computation or storage at the servers \cite{LMYA, YaYW, WCJ,WSJC2}.

This same fundamental tradeoff lies at the core of our study, where we examine the $(K,L,M)$ distributed linearly separable function computation problem over a finite\footnote{We will show later on that many of our results hold for any field including the field of real and complex numbers.} field \(\mathbb{F}_q\). This problem motivates us as it arises in many domains such as signal processing, machine learning, control theory \cite{VTr, SDa, NGSA}, etc. In this setting, a user requests $L$ function values from a master node, each corresponding to the output of an independent function $F_\ell, \ell \in [L],$ acting on a dataset library $\mathcal{W}$ consisting of $K$ independent and identically distributed (i.i.d.) datasets. As we clarify later, these functions are linearly separable over $K$ basis subfunctions $f_j, j \in [K]$ of the datasets, enabling the master node to parallelize the computations across $N$ distributed servers. Each server is assigned a subset of no more than $M$ chosen subfunctions to compute, each evaluated on a dataset stored by that server. Each server is thus said to enjoy a finite computational capability $M$. After completing their computations, the servers transmit linear combinations of their computed values to the user, who then must recover the $L$ requested function outputs. Thus, the objective, for any given $(K,L,M)$ instance of the problem, is to design a distributed computing scheme that best allocates tasks across the servers, and best provides communication schemes, in order to minimize the communication cost needed for delivering the desired functions.

To put this objective into context, and to get a sense of the corresponding tradeoff between computation ($M$) and communication, we quickly note that as one might expect, having $M=K$ maps to a centralized setting, where the single server makes $L$ transmissions to serve the $L$ functions. On the other extreme, having $M=1$ corresponds to a scenario where each server can compute only one subfunction, thus forcing $N=K$ transmitting servers, and a maximum communication cost of $K$. We aim to identify and meet the optimal tradeoff in between these two extreme scenarios, by properly assigning the subfunctions to the servers in a way that minimizes the total amount of data transferred between the servers and the user.

\subsection{Related Works}

Our setting enjoys numerous connections with various themes of distributed computing. We proceed to mention some of these connections, while also clarifying key differentiating factors. 

Our work is most closely related to \cite{WSJC,WSJC2}, which study a similar problem of linearly separable function computation with $N$ servers, a single user, multiple requested functions over $K$ datasets, and an objective of minimizing communication cost for a given $M$.  The main distinction with our work is that \cite{WSJC,WSJC2} emphasize on straggler mitigation, and on the specific case of the cyclic task assignment; based on this cyclic structure -- and always with a focus on accounting for stragglers -- they proceed to design novel communication schemes as well as useful information-theoretic converses.   Related work can also be found in the more recent~\cite{KhE,KhE2}, which consider the \emph{multi-user version} of the problem in \cite{WSJC} without stragglers, where now each server is connected to multiple, but not all, users, each asking for their own desired function. Both these works~\cite{KhE,KhE2} aim to reduce the communication and computation cost, the first by exploiting the properties of covering codes, and the second by employing optimal tilings from tessellation theory. In this multi-user setting, the first work provides optimality under the restrictive assumption that all possible demands are simultaneously requested by the users, and the second work -- once translated onto our single-user setting -- operates under various restrictive assumptions of disjoint tasks across servers. These assumptions do not appear in our work. 

Our problem also has strong natural links to \emph{distributed source coding for function computation}, which often relates to our setting; Multiple sources (that play the role of servers with access to datasets) communicate via parallel channels to a destination (our user) who wishes to compute various functions of the data at the sources. Such works (cf.~\cite{KoM,SlW,HaK,Wag,LPVKP}) make significant progress on the challenging case of correlated source data (which, in our setting, corresponds to datasets or task outputs), by designing efficient compression schemes that exploit these correlations while taking into account the structure of the requested functions. In particular,  various works focus on identifying the \emph{function \textcolor{black}{computation} capacity}, corresponding to the maximum rate at which distributed encoders (servers) compress their observations so that a decoder recovers a function of sources. Key works include \cite{AFKZ}, which introduced cut-set bounds for independent sources, and \cite{GYYL}, which refined the approach to yield tighter bounds. More recent works include \cite{GuZ} on the two-source arithmetic sum under channel asymmetry, and \cite{FeM} (see also \cite{MSSE}) on multiple sources and general functions. Another related direction concerns the compression of vector-linear functions. For instance, \cite{LFPNG} studies the problem where a single user aims to recover multiple linear combinations of the sources (losslessly or with loss), designing schemes based on nested linear codes, while \cite{GSZ} provides capacity characterizations under specific source–decoder connectivity, emphasizing on specific instances (e.g., $K=3$, $N=2$). 
The related literature is broad, and while our overview is necessarily selective, all prior approaches remain fundamentally different from ours; In addition to the fact that such studies often (though not always) focus on the special case $M=1$ (one dataset per source/server), the key distinction from our work lies in source design: in distributed source coding the sources are fixed, whereas in our setting they are shaped through careful task assignment. It is this key overarching distinction, together with having larger $M$, that introduces challenging combinatorial optimization elements in designing schemes and converses, which we address in our study.

Another related and well-studied direction is \emph{Coded MapReduce} \cite{LMYA}, which focuses on distributed computation of separable functions under a framework where servers compute intermediate function values and exchange these, via coded messages, to obtain the inputs needed for their final assigned tasks. The challenge here is to jointly design the task assignment and internode communication strategy in order to best mitigate the communication bottleneck. Variants have been proposed to handle stragglers \cite{CCW}, and to account for predefined task assignments \cite{WaW}, heterogeneous channels \cite{BWW, PNR}, and different network connectivities \cite{BrEl}, often borrowing techniques from device-to-device coded caching. Related problems of function retrieval with communication minimization have also been studied in cache-aided broadcast systems \cite{WSJTC} and broadcast networks with side information \cite{YaJ, MaT}. While the Coded MapReduce setting carries certain similarities to our setting, such as for example the need to carefully assign datasets and tasks across the servers, the unifying distinction between these works and our line of work -- beyond a deep dependence on subpacketization, the fundamentally different topology and the presence of data-exchange among the servers -- is that they critically rely on side information and focus on its use in alleviating the rank-one communication bottleneck.

Further connections to our setting can be found in the rich line of work that studies the application of distributed computing in certain high-impact tasks. Associated to our specific case $L=1$ (a single requested function), this includes \emph{gradient coding} \cite{YAb,TLDK,RTTD}, which develops techniques to mitigate stragglers in distributed learning. Here, the master computes a single gradient with the help of $N$ servers ($N$ depends on $K$, $M$, and the number of tolerated stragglers), each evaluating partial gradients. Most schemes use cyclic or fractional repetition task assignment strategies \cite{YAb,TLDK,RTTD}, with additional 
variants to be found in~\cite{RPPA,LKAS,OUG,GJWC}. 
Another related direction -- motivated in part by the growing computational and communication demands of machine learning -- concerns \emph{distributed matrix–vector and matrix–matrix multiplication}, with the majority of works emphasizing on straggler-prone settings \cite{LLPPR, RTV, LSR, YMA, DFHJCG}. In these works, the input matrices are encoded, and subsets of the encoded submatrices are assigned to servers for computation. Most existing schemes adopt either erasure-coding techniques or polynomial-based methods for encoding. Related studies \cite{DCG, WLSY} focus on computing linear transforms of high-dimensional vectors under both straggler effects and explicit computational load constraints (see also the survey \cite{RDT}). Several extensions have incorporated additional privacy and security guarantees as well~\cite{YLRKSA, XBW, YuS, JiJ, ChR, MLG, DEK, MaH, HBW}. A key distinction from our work is the highly specific nature of the target function and the primary focus on mitigating stragglers, which typically directs the problem toward coded task assignments inspired by error-correcting codes.

\subsection{Contributions}
In this work, we study distributed linearly separable function computation with the goal of minimizing the total communication cost $R$ from servers to the user. In particular, the user requests $L$ linearly separable functions of the $K$ basis subfunctions ${f_j(W_j): j \in [K]}$, captured by a demand matrix $\mathbf{D} \in \mathbb{F}_q^{L \times K}$, where $\mathcal{W}=\{W_1,\dots,W_K\}$ is the dataset library at the master node. The challenge arises from the limited computating capability of each server, which can compute at most $M\leq K$ subfunctions. We seek the fundamental limits of the communication–computation ($R$–$M$) tradeoff for arbitrary $K$ and $L$, under the assumption that servers transmit only linear combinations of their computed subfunctions, without subpacketization. 
The problem reduces to jointly designing (i) the dataset/task assignment under a strict per-server budget $M$, and (ii) the transmissions that allow the user to recover all requested functions without error. As we clarify later on, performance will be measured by the worst-case communication cost over all possible demands.

The technical contributions of this work are summarized in the following.   

\begin{itemize}
    \item In terms of achievability, we first introduce a novel technique that provides a sufficient condition for designing both task assignments and server transmissions. The core idea (Lemma~\ref{lemma1}) leverages designed left nullspaces of carefully selected submatrices of the demand matrix to guarantee exact decodability when $K \leq L+M-1$. In this regime, we construct a `nullspace-based matrix' (cf.~\eqref{eq:NullMat}) for any task assignment to test feasibility, and, if feasible, the same matrix specifies the transmissions. Moreover, for all $K,L,M$ satisfying $K \leq L+M-1$, we explicitly provide feasible assignments and transmissions achieving the globally optimal communication cost. While prior works \cite{WSJC,WSJC2} employed nullspace ideas to design transmissions, our contribution extends this principle to task assignment itself. The defining distinction of our approach is that the nullspace design dictates the task allocation, in contrast to \cite{WSJC,WSJC2}, where cyclic assignment predetermined the nullspace structure. This extension not only allows for reaching the global optimal over all task assignments, but also ensures exact decodability of all demands, going beyond the probabilistic guarantees in \cite{WSJC,WSJC2}\footnote{Indeed, for small field sizes $q$, those schemes may fail on a significant fraction of demand matrices (see Table I in \cite{WSJC2}).}.
    
A few additional design elements appear in the subsequent Schemes 1 and 2, designed for the case of \(K>L+M-1\).
    \begin{itemize}
        \item Scheme 1 (Theorem~\ref{thm1}) applies a column-wise partition of the demand matrix into submatrices, on each of which the nullspace-based design for task assignment and transmissions is applied independently. We further identify the optimal partition that minimizes communication cost. The scheme works for any $(K,L,M)$, and the corresponding cost takes the form $R = \min(K, L\lceil K/(L+M-1)\rceil)$.
        
        \item Scheme 2 (Theorem~\ref{thm2}) is distinguished by admitting a task assignment independent of the demand matrix. It applies in the regime of $M \geq K/2$, where performance is paradoxically improved by augmenting the demand matrix with carefully designed virtual demands, enabling a more efficient fusion of the partitioning and nullspace approaches. To further broaden its applicability, the scheme also incorporates a row-wise partitioning technique. 
    \end{itemize}
    \item A key contribution of this work is the development of new converses for distributed linearly separable function computation. We first establish an information-theoretic converse via a novel degrees-of-freedom argument on sparse matrix factorizations, showing Scheme 1 to be within a small constant factor from the optimal. Furthermore, we present an additional converse that, for the first time, identifies the general covering number from combinatorial design theory as a fundamental limiting factor for communication cost in distributed computing. Furthermore, our work provides a new class of covering designs, that yield even tighter converses. All together, these new connections yield converses that are proven tight, matching the performance of our designed achievability schemes, thus revealing a surprising and deep structural equivalence.  
    \begin{itemize}
        \item To achieve the above, we first reformulate our problem as a sparse matrix factorization problem  $\mathbf{D}_{L\times K}=\mathbf{C}_{L\times R}\mathbf{A}_{R\times K}$ where the challenge is to minimize the inner dimension $R$ under the constraint that each row of $\mathbf{A}$ has at most $M$ non-zero elements. This formulation yields a novel information-theoretic lower bound (Theorem~\ref{th:conv}) on the communication cost of the $(K,L,M)$ problem, derived under the no-subpacketization assumption. The converse stems from the idea that, in the factorization \(\mathbf{D}=\mathbf{C}\mathbf{A}\), the total degrees of freedom available in choosing \(\mathbf{A}\) must exceed those in \(\mathbf{C}\) for a given demand matrix \(\mathbf{D}\). These degrees of freedom are then bounded using combinatorial counting arguments and matrix entropic inequalities to obtain the lower bound. This derived lower bound is subsequently used to show that under basic divisibility conditions, the communication cost of Scheme 1 is within a factor of \(2\) from the optimal when the underlying field size \(q\geq eK/M\), while when the divisibility condition is not met, the optimality gap is at most \(3\) (cf. Theorem~\ref{thm:gap}).
        \item In Theorem~\ref{thmcovering}, we derive another lower bound for the \((K,L,M)\) distributed linearly separable function computation problem. This bound establishes a novel connection between the distributed computing setting and combinatorial design theory. Specifically, for given \(K\), \(L\), and \(q>K\), we show that the existence of a novel combinatorial structure, here called a \emph{multi-level covering design\footnote{This new design will be a restricted version of the classical covering design.}} with appropriate parameters, is necessary for the existence of a sufficiently sparse matrix factorization for some \(\mathbf{D}\in \mathbf{F}_q^{L\times K}\), and thus is necessary for any given rate $R$ to be achievable. This relation allows the communication cost of the distributed computing problem to be lower bounded by the multi-level covering number of the associated design, as well as by the general covering number of classical general covering designs. Although determining the multi-level and general covering numbers is difficult, we derive a lower bound for the case $L=2$ 
        using non-trivial counting arguments, which, interestingly, coincide with the communication cost of Scheme 1 when \((M+1)|K\).
    \end{itemize}
    \item Of importance is also Corollary~\ref{cor:RealField} which clarifies that the achievable schemes described in Theorem~\ref{thm1} and Theorem~\ref{thm2}, as well as the converse in Theorem~\ref{thmcovering}, remain valid even when the computation is performed over any field, including the field of real numbers. 
    \item Finally, in addition to the main results, we also study the increase in communication rate when, for any fixed $M$, we reduce the number of computing servers. Our tradeoff analysis, reflected in Theorem~\ref{thm6}, shows that under our assumptions, the communication rate penalty remains small across a wide range of parameters, while the reduction in server resources is significant. Viewed in reverse, this means that nearly the same cumulative communication cost can be sustained even with substantially reduced overall computing power, simply by using fewer servers.
        
\end{itemize}

\subsection{Organization}
Section~\ref{sec:sysmodel} formally defines the distributed computing problem of linearly separable functions. Section~\ref{sec:mainres} describes the achievable schemes proposed for different settings, and Section~\ref{sec:converse} presents the converse bounds for the communication cost achieved by a distributed linearly separable function computing scheme. Before we conclude the paper, Section~\ref{RvsN} explores the tradeoff between the communication cost and the number of servers used in the system. 
To help the reader follow the exposition, we intersperse various examples throughout the paper.

\subsection{Notations and Definitions}
For a positive integer $n$, we let $[n] = \{ 1,2, \ldots,n\}$. For $m$,  $n \in \mathbb{Z}^{+}$ such that $m <n$, $[m:n]$ denotes the set $\{m, m+1,\ldots,n\}$, and $m \mid n$ denotes $m$ divides $n$. 
All vectors are assumed to be column vectors. For a vector $\mathbf{x}$, the number of non-zero entries in $\mathbf{x}$ is denoted by $||\mathbf{x}||_0$. 
A vector $\mathbf{x}$ is said to be $k-$sparse if $||\mathbf{x}||_0 \leq k$.
Binomial coefficients are denoted by $\binom{n}{k}$, where $\binom{n}{k} = \frac{n!}{k!(n-k)!}$ and $\binom{n}{k}=0$ for $n <k$. For $N,k \in \mathbb{Z}^+$, $\binom{[N]}{k}$ denotes the set of all $k-$sized subsets of $[N]$. For any set $\mathcal{S}$, the cardinality of $\mathcal{S}$ is denoted by $|\mathcal{S}|$. The complement of set $\mathcal{S}$ is denoted by $\mathcal{S}^c$. The empty set is denoted by $\varnothing$. 
The vertical concatenation of two matrices $\mathbf{A}_{m_1 \times n}$ and $\mathbf{B}_{m_2 \times n}$ are denoted by $[\mathbf{A};\mathbf{B}]$. The $i$-th row and the $j$-th column of a matrix $\mathbf{A}_{m \times n}$ are denoted by $\mathbf{A}(i,:)$ and $\mathbf{A}(j,:)$, respectively. For a set $\mathcal{S} \subseteq [n]$ and a matrix $\mathbf{A}$ with $n$ columns, $\mathbf{A}_{\mathcal{S}}$ denotes the submatrix of $\mathbf{A}$ formed by the columns indexed by $\mathcal{S}$. 
For any $x \in \mathbb{R}^+$, $\ceil{x}$ denotes the smallest positive integer greater than or equal to $x$, and $\floor{x}$ denotes the largest positive integer less than or equal to $x$. The finite field of $q \geq 2$ elements is denoted by $\mathbb{F}_q$. A vector of length $n$ with all zeros is denoted by $\mathbf{0}_{n \times 1}$. An $n-$length unit vector with a one at the $i$-th position and zeroes elsewhere is denoted by $\mathbf{e}_i$. The support of a vector $\mathbf{x} \in \mathbb{F}_q^n$ is defined as $supp(\mathbf{x})=\{ i \in [n]: x_i \neq 0\}$.

\section{System Model and Problem Formulation}
\label{sec:sysmodel}

\begin{figure}
\begin{center}   
\includegraphics[width=10cm]{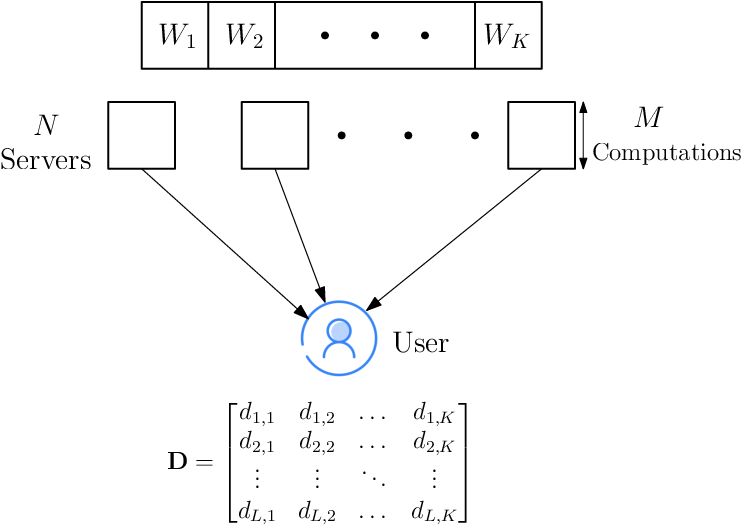}
\caption{System model: \((K,L,M)\) distributed linearly separable function computation.}
\label{SystemModel}
\end{center}
\end{figure}

Consider a distributed computing system with a master node having a library of $K$ i.i.d. datasets $\mathcal{W}=\{W_1,W_2,\ldots,W_K\}$, where $W_k \in \mathbb{F}_q^B$, $\forall k \in [K]$.\footnote{\textcolor{black}{We will for now consider any $q$. As we will see later on, the achievable schemes are applicable for any field, finite or not.  The converse of Theorem~\ref{th:conv} is applicable for any finite field \(\mathbf{F}_q\), while the converse in Theorem~\ref{thmcovering} presented later requires \(q > K\) and also applies over the field of real and complex numbers. The optimality gap of Theorem~\ref{thm:gap} (based on Theorem~\ref{thm1} and Theorem ~\ref{th:conv}), holds for all finite fields. The gap reduces to $2$ or $3$ for any $q\geq eK/M$.}}
In this setting, a single user wishes to compute \(L\) independent functions \(F_1, F_2, \ldots, F_L\) on \(\mathcal{W}\), where
\(
F_\ell : (\mathbb{F}_q^B)^K \rightarrow \mathbb{F}_q^T,~ \forall \ell \in [L]
\)
and \(T\) denotes the size of each function value. Furthermore, each function $F_\ell$ is linearly decomposable as
\begin{align*}
	 F_\ell(\mathcal{W})= d_{\ell,1}f_{1}(W_1)+d_{\ell,2}f_{2}(W_2)+\cdots +d_{\ell,K}f_{K}(W_K), \quad \forall \ell \in [L]
\end{align*}
where $d_{\ell,j} \in \mathbb{F}_q$, and where the subfunction $f_{j}: \mathbb{F}_q^B \rightarrow \mathbb{F}_q^T $, $j \in [K]$,  can be linear or non-linear and can be computationally intensive. Consequently, the $L$ requested functions can be represented as 
\begin{align}
 \begin{bmatrix}
    F_1(\mathcal{W}) \\
    F_2(\mathcal{W}) \\
    \vdots  \\
  F_L(\mathcal{W})\\
 \end{bmatrix} =
\underbrace{ \begin{bmatrix}
 d_{1,1}  &  d_{1,2}  & \dots  & d_{1,K} \\
  d_{2,1}  & d_{2,2}  & \dots  & d_{2,K} \\
  \vdots & \vdots & \ddots & \vdots \\
  d_{L,1} & d_{L,2} & \dots & d_{L,K}
 \end{bmatrix}}_{\mathbf{D}}
 \begin{bmatrix}
   f_1(W_1) \\
   f_2(W_2) \\
   \vdots \\
   f_K(W_K) \\
   \end{bmatrix}
  \label{eq:rep}
\end{align}
where -- under the worst-case assumption -- we consider the case of $\text{rank}_q(\mathbf{D})=L$.
Since the $L$ functions are linearly separable, the computation of the $F_\ell(\mathcal{W})$'s can be performed in a distributed manner across the $N$ servers. We consider the case where the computational capability of the servers is limited, in that each server can compute a maximum of $M \textcolor{black}{\leq} K$ subfunctions. This model is illustrated in Figure~\ref{SystemModel}. 

The system operates in three phases: the \textit{demand phase}, the \textit{computing phase}, and the \textit{communication phase}.

\textit{Demand phase}: In the demand phase, the user informs the master node of its $L$ desired functions, which are naturally described by the matrix $\mathbf{D} \in \mathbb{F}_q^{L \times K}$.

\textit{Computing phase}: Subject to the constraint that each server computes at most $M$ subfunctions, the master uses the demand matrix $\mathbf{D}$ to assign to server $n$ a subset $\mathcal{M}_n \subset [K]$ of subfunctions and the associated datasets. At the end of this phase, server $n$ evaluates $f_j(W_j) \in \mathbb{F}_q^T$ for all $j \in \mathcal{M}_n$. We often refer to the collection \(\mathcal{M} = \{\mathcal{M}_n : n \in [N]\}\) as the \emph{task assignment}.

\textit{Communication phase}: After computing, server $n$ sends $r_n$ linearly encoded messages based on its local outputs $f_j(W_j)$, where each message is of the form\footnote{We adopt the standard no-subpacketization assumption, where servers do not divide subfunction outputs into smaller parts.
}
\begin{equation} \label{eq:xnr} x_{n,r} = \sum_{j \in \mathcal{M}_n}\alpha_{n,r,j} f_j(W_j) \end{equation}  for $r \in [r_n]$, with coefficients $\alpha_{n,r,j} \in \mathbb{F}_q$. This means that each server $n$ transmits $ \mathbf{x}_n= [x_{n,1}, x_{n,2}, \ldots, x_{n,r_n}]^{\intercal}$, where 
\begin{align}
 \begin{bmatrix}
   x_{n,1} \\
   x_{n,2}\\
    \vdots  \\
  x_{n,r_n}\\
 \end{bmatrix} =
\underbrace{ \begin{bmatrix}
 \alpha_{n,1,1}  &  \alpha_{n,1,2}  & \dots  & \alpha_{n,1,K} \\
  \alpha_{n,2,1}  & \alpha_{n,2,2}  & \dots  & \alpha_{n,2,K} \\
  \vdots & \vdots & \ddots & \vdots \\
  \alpha_{n,r_n,1} & \alpha_{n,r_n,2} & \dots & \alpha_{n,r_n,K}
 \end{bmatrix}}_{\mathbf{A}_n \in \mathbb{F}_q^{r_n \times K}}
 \begin{bmatrix}
   f_1(W_1) \\
   f_2(W_2) \\
   \vdots \\
   f_K(W_K) \\
   \end{bmatrix} \ .
   \label{eq:trans1}
\end{align}
Since $x_{n,r}$ is a linear combination of at most $M$ subfunctions, then  $||\mathbf{A}_n(r,:)||_0 \leq M, \forall r \in [r_n] \text{ and } n \in [N]$.
For $\mathbf{x}=[\mathbf{x}_1; \mathbf{x}_2; \ldots; \mathbf{x}_N]$ denoting the entire set of transmissions across all the servers, for $R = \sum_{n \in [N]}r_n$, and for $\mathbf{A}=[\mathbf{A}_1; \mathbf{A}_2; \ldots; \mathbf{A}_N] \in \mathbb{F}_q^{ R \times K}$ denoting the \emph{encoding matrix}, then the entire set of transmissions can be represented as   
$$\mathbf{x}=\mathbf{A}[f_1(W_1), f_2(W_2), \ldots, f_K(W_K)]^{\intercal}$$ where each row of $\mathbf{A}$ has a non-empty support of cardinality at most $M$.

The communication link between each of the $N$ servers and the user is assumed to be error-free and non-interfering. The user is allowed to perform linear operations on the received messages, in order to decode the $L$ function values, and thus this decoding operation can be represented as 
\begin{align}
 \begin{bmatrix}
    F_1(\mathcal{W}) \\
    F_2(\mathcal{W}) \\
    \vdots  \\
  F_L(\mathcal{W})\\
 \end{bmatrix} =
\underbrace{ \begin{bmatrix}
 c_{1,1}  &  c_{1,2}  & \dots  & c_{1,R} \\
  c_{2,1}  & c_{2,2}  & \dots  & c_{2,R} \\
  \vdots & \vdots & \ddots & \vdots \\
  c_{L,1} & c_{L,2} & \dots & c_{L,R}
 \end{bmatrix}}_{\mathbf{C} \in \mathbb{F}_q^{L \times R}}
 \underbrace{\begin{bmatrix}
   \mathbf{x}_1 \\
  \mathbf{x}_2 \\
   \vdots \\
  \mathbf{x}_N  \\
   \end{bmatrix}}_{\mathbf{x} \in \mathbb{F}_q^{R \times 1}} = \mathbf{C} \mathbf{A}
   \begin{bmatrix}
   f_1(W_1) \\
   f_2(W_2) \\
   \vdots \\
   f_K(W_K) \\
   \end{bmatrix}
   \label{eq:dec}
\end{align}
where the matrix $\mathbf{C}$ is composed of the combining coefficients. From \eqref{eq:rep} and \eqref{eq:dec}, and from the fact that the scheme must hold independent of the dataset realizations and subfunction outputs, 
it is evident that the demand matrix $ \mathbf{D} \in \mathbb{F}_q^{L \times K}$ must be decomposed as
\begin{align}
\label{eq:TransDCA}
 \mathbf{D}=\mathbf{C}\mathbf{A}.
\end{align}
Because of this equivalence between the distributed computing problem and matrix decomposition, our aim becomes that of designing two matrices $\mathbf{C} \in \mathbb{F}_q^{L \times R  } $ and $\mathbf{A} \in \mathbb{F}_q^{R \times K  }$, where $R \geq L$, that factorize $\mathbf{D}$ subject to an $M$-sparsity constraint on every row of $\mathbf{A}$.

The communication cost here represents the total number of transmissions, from the servers to the user, required for recovery of the $L$ desired function outputs at the user. 
In particular, for a given $\mathbf{D} \in \mathbb{F}_q^{L \times K}$, the communication cost takes the form
\begin{equation*}
 R_{\mathbf{D}}(M) = \sum_{n \in [N]}r_n 
\end{equation*}
where we assume unit-length file messages, i.e., we assume $T = |x_{n,r}| = 1, \ n \in [N],r \in [r_n]$.\footnote{A subfunction output/coded transmission comprising $T$ symbols is treated as one unit.} 
The rate of interest in our work will represent the worst-case communication cost 
$$R(K,L,M) = \max_{\mathbf{D} \in \mathbb{F}_q^{L \times K}} R_{\mathbf{D}}(M) $$
over all full-rank matrices $\mathbf{D}$, and thus our interest is in identifying the optimal rate
\begin{equation*}
  R^{*}(K,L,M) = \inf\{R(K,L,M): R(K,L,M) \text{ is achievable} \}
\end{equation*}
where the infimum is over all task assignments and linear transmission and decoding policies. Our objective is to design the assignment and communication scheme that approaches this optimum.

\section{Achievability}
\label{sec:mainres}
In this section, we present two achievable schemes for the distributed linearly separable function computation problem. The first scheme (Theorem~\ref{thm1}) applies to all values of \(M,K\), while the second scheme (Theorem~\ref{thm2}) is designed to improve the performance of the first scheme in the region of $M\geq K/2$. Additionally, this second scheme enjoys a demand-agnostic task assignment. We remind the reader that the parameter range of interest is naturally that for which both $M$ and $L$ remain less than $K$.

\subsection{Scheme 1: A General Achievability Scheme}
\begin{thm}
    \label{thm1}
    For the $(K,L,M)$ distributed linearly separable function computation problem, the rate
    \begin{equation}
    \label{eq:thm1}
        R_1(K,L,M) = \min\left\{K,L\left\lceil{\frac{K}{L+M-1}}\right\rceil\right\}
    \end{equation}
    is achievable. 
\end{thm}
\textit{Proof.} The proof of the achievability of the rate expression in \eqref{eq:thm1} is divided into two cases. Case~1 for \(K\leq L+M-1\), and Case~2 for \(K > L+M-1\). The crux of the proof of Case~1 lies in Lemma~\ref{lemma1} presented next, while for Case~2, we combine the idea presented in Lemma~\ref{lemma1} with a demand matrix partitioning technique. 
\subsubsection{\textbf{Case 1 } (\(K\leq L+M-1\))}
\label{Case1}~\\
\textcolor{black}{We first consider the case \(K \leq L + M - 1\). When doing so, it is sufficient to focus on the scenario \(K = L + M - 1\), because if \(K < L + M - 1\), each server uses only \(M' = K - L + 1 < M\) of its computing capability, reducing the problem to the case \(K = L + M' - 1\).} Consider a demand matrix \(\mathbf{D}\in\mathbb{F}_q^{L\times K}\), where \(\mathbf{D} =[\mathbf{d}_1,\mathbf{d}_2,\dots,\mathbf{d}_K]\). Recall that \(\mathcal{M}_n\subset [K]\) is the set of indices of datasets known to server \(n\), where \(|\mathcal{M}_n|=M = K-L+1\) for every \(n\in [N]\). We now define a submatrix  
\begin{equation*}
    \mathbf{D}_{\mathcal{M}_n} = [\mathbf{d}_j:j\in \mathcal{M}_n],\quad n\in [N]
\end{equation*}
of the demand matrix \(\mathbf{D}\) by choosing the columns indexed with the set \(\mathcal{M}_n\). Consider the following matrix 
\begin{equation}
\label{eq:NullMat}
    \mathbf{N}_{\mathbf{D},\mathcal{M}} = [\mathcal{N}(\mathbf{D}_{\mathcal{M}_1^c}^\intercal),\mathcal{N}(\mathbf{D}_{\mathcal{M}_2^c}^\intercal),\dots,\mathcal{N}(\mathbf{D}_{\mathcal{M}_N^c}^\intercal)]^\intercal
\end{equation}
where \(\mathbf{D}_{\mathcal{M}_n^c}\) denotes the matrix obtained by choosing the columns of \(\mathbf{D}\) indexed with the set \([K]\backslash\mathcal{M}_n\) and where \(\mathcal{N}(.)\) is the nullspace operator, which outputs a set of basis vectors of the nullspace of the matrix on which the operator acts. For any \(\mathcal{M}_n\) with \(|\mathcal{M}_n|=M=K-L+1\), the submatrix \(\mathbf{D}_{\mathcal{M}_n^c}\) is of size \(L\times (L-1)\), and therefore \(\mathbf{D}_{\mathcal{M}_n^c}\) has a non-trivial left nullspace. Consequently, the rank of the left nullspace of \(\mathbf{D}_{\mathcal{M}_n^c}\) is at least one, and \(\mathcal{N}(\mathbf{D}_{\mathcal{M}_n^c}^\intercal)\) contains at least one vector for every \(n\in [N]\). Thus, the number of rows in the matrix \(\mathbf{N}_{\mathbf{D},\mathcal{M}}\) is at least \(N\). We now have the following lemma. 

\begin{lem}
\label{lemma1}
    For a given \(\mathbf{D}\in \mathbb{F}_q^{L\times K}\) with \(\text{rank}_q(\mathbf{D})= L\), and \(\mathcal{M}=\left\{\mathcal{M}_1,\mathcal{M}_2,\dots,\mathcal{M}_N\right\}\)
    , the rate \(L\) is achievable if 
    \begin{equation}
    \label{eq:nullspace}
        \text{rank}_q(\mathbf{N}_{\mathbf{D},\mathcal{M}})=L.
    \end{equation}
    Furthermore, for any given \(K,L\), and \(M\) with \(K=L+M-1\), and for any \(\mathbf{D}\in \mathbb{F}_q^{L\times K}\), the rate 
    \begin{equation*}
        R_1(K,L,M)=L
    \end{equation*}
    is achievable. 
\end{lem}
\begin{proof}[Proof of Lemma~\ref{lemma1}]
    Consider a demand matrix \(\mathbf{D}\in \mathbb{F}_q^{L\times K}\), where  \(\text{rank}_q(\mathbf{D})=L\). Furthermore, assume that the task assignment \(\mathcal{M}=\left\{\mathcal{M}_1,\mathcal{M}_2,\dots,\mathcal{M}_N\right\}\) is such that the condition
        \begin{equation*}
        \text{rank}_q(\mathbf{N}_{\mathbf{D},\mathcal{M}})=L
    \end{equation*}
    holds. Then, the matrix \(\mathbf{N}_{\mathbf{D},\mathcal{M}}\) has at least a set of \(L\) independent rows. Let \(\tilde{\mathbf{N}}_{\mathbf{D},\mathcal{M}}\) be an \(L\times L\) matrix constructed by choosing a set of \(L\) independent rows from \({\mathbf{N}}_{\mathbf{D},\mathcal{M}}\). In the communication phase, there is a transmission corresponding to each row in \(\tilde{\mathbf{N}}_{\mathbf{D},\mathcal{M}}\). Let the \(\ell\)-th row of \(\tilde{\mathbf{N}}_{\mathbf{D},\mathcal{M}}\), denoted by \(\tilde{\mathbf{N}}_{\mathbf{D},\mathcal{M}}(\ell,:)\), corresponds to the task assignment \(\mathcal{M}_n\in \mathcal{M}\) for some \(n\in [N]\). Then, server \(n\) makes the following transmission 
    \begin{equation}
    \label{eq:trans}
         \tilde{\mathbf{N}}_{\mathbf{D},\mathcal{M}}(\ell,:)\mathbf{D}[f_1(W_1),f_2(W_2),\dots,f_K(W_K)]^\intercal.
    \end{equation}
    Now, it remains to show the following two claims.
     \begin{itemize}
         \item \textit{Claim 1:} Server \(n\) can construct the transmission in \eqref{eq:trans} from the available datasets in \(\mathcal{M}_n\). 
         \item \textit{Claim 2:} The user can decode the \(L\) required functions using the \(L\)  transmissions received from the servers.
     \end{itemize}
     \textit{Claim 1} essentially means that the message in \eqref{eq:trans} is a linear combination of the subfunctions \(\{f_j(W_j):j\in \mathcal{M}_n\}\). In other words, we need to show that
     \begin{equation}
     \label{eq:claim1}
     supp (\tilde{\mathbf{N}}_{\mathbf{D},\mathcal{M}}(\ell,:)\mathbf{D})\subseteq \mathcal{M}_n.
     \end{equation}
     However notice that 
     \begin{equation*}
         \tilde{\mathbf{N}}_{\mathbf{D},\mathcal{M}}(\ell,:) \in \mathcal{N}(\mathbf{D}_{\mathcal{M}_n^c}^\intercal)
     \end{equation*}
     and thus that 
        \begin{equation*}
         \tilde{\mathbf{N}}_{\mathbf{D},\mathcal{M}}(\ell,:) \mathbf{D}_{\mathcal{M}_n^c} = \mathbf{0}_{1\times K-M}
     \end{equation*}
     which implies that \eqref{eq:claim1} is true, and thus Claim 1  follows.

     From \eqref{eq:trans}, the \(L\) transmitted messages can be written in matrix form as
     \begin{equation}
         \label{eq:Ltrans}
         \mathbf{x}_{L\times 1} =\tilde{\mathbf{N}}_{\mathbf{D},\mathcal{M}} \  \mathbf{D}[f_1(W_1),f_2(W_2),\dots,f_K(W_K)]^\intercal.
     \end{equation}
    Since \(\text{rank}_q(\tilde{\mathbf{N}}_{\mathbf{D},\mathcal{M}})=L\), the matrix \(\tilde{\mathbf{N}}_{\mathbf{D},\mathcal{M}}\) is invertible, and thus the user can decode the requested functions from \(\mathbf{x}\) as
          \begin{equation}
         \label{eq:Decode}
         \mathbf{D}[f_1(W_1),f_2(W_2),\dots,f_K(W_K)]^\intercal= (\tilde{\mathbf{N}}_{\mathbf{D},\mathcal{M}})^{-1}\mathbf{x}_{L\times 1}.
     \end{equation}

     To show the achievability of \(R_1(K,L,M)=L\) for any given \(\mathbf{D}\in \mathbb{F}_q^{L\times K}\) with \(K=L+M-1\), we give an explicit construction of \(\mathcal{M}=\{\mathcal{M}_1,\mathcal{M}_2,\dots,\mathcal{M}_L\}\) satisfying the condition \(\text{rank}_q(\mathbf{N}_{\mathbf{D},\mathcal{M}})=L\). Note that here we have $N=L$. Without loss of generality, we assume that \(\text{rank}_q(\mathbf{D})=L\). If this is not the case, we can replace some \(L-\text{rank}_q(\mathbf{D})\) linearly dependent rows of \(\mathbf{D}\) with a different set of \(L-\text{rank}_q(\mathbf{D})\) vectors, so that the resulting matrix has full row rank. If the user can decode the new functions, the removed \(L-\text{rank}_q(\mathbf{D})\) functions can still be recovered as linear combinations of the original \(\text{rank}_q(\mathbf{D})\) functions that were retained. From the full-row-rank matrix \(\mathbf{D}\), we find an invertible submatrix \(\mathbf{D}_\mathcal{L}\) of size \(L\times L\), where \(\mathcal{L}= \{i_1,i_2,\dots,i_L\}\) denotes the index set of columns in \(\mathbf{D}\) chosen to form \(\mathbf{D}_\mathcal{L}\). Then, we form the task assignment set
     \begin{equation}
     \label{eq:dataset1}
         \mathcal{M}^\star = \{\mathcal{M}_1^\star,\mathcal{M}_2^\star,\dots,\mathcal{M}_L^\star\} 
     \end{equation}
     where, for every \(\ell\in [L]\)
     \begin{equation}
         \label{eq:datasetassign}
         \mathcal{M}_\ell^\star = \left\{i_\ell\right\} \bigcup \left(\left[K\right]\backslash \mathcal{L}\right).
     \end{equation}
     Note that, \(|\mathcal{M}_\ell^\star|=1+K-L=M\) for every \(\ell\in [L]\). To complete the proof of achievability of the rate \(R_1(K,L,M)=L\), it remains to verify that \(\text{rank}_q(\mathbf{N}_{\mathbf{D},\mathcal{M}^\star})=L\). Suppose  \(\text{rank}_q(\mathbf{N}_{\mathbf{D},\mathcal{M}^\star})<L\), then there exists a row \(\mathbf{N}_{\mathbf{D},\mathcal{M}^\star}(\ell',:)\) of \(\mathbf{N}_{\mathbf{D},\mathcal{M}^\star}\) which can be represented as 
     \begin{equation}
     \label{eq:contra0}
         \mathbf{N}_{\mathbf{D},\mathcal{M}^\star}(\ell',:) = \sum_{\ell=1,\ell\neq \ell'}^L \alpha_\ell \mathbf{N}_{\mathbf{D},\mathcal{M}^\star}(\ell,:)
     \end{equation}
     where \(\alpha_\ell\neq 0\) for some \(\ell\). On one hand, \(\mathbf{N}_{\mathbf{D},\mathcal{M}^\star}(\ell',:)\) is a vector in the left nullspace of \(\mathbf{D}_{\mathcal{M}_{\ell'}^c} = \mathbf{D}_{\mathcal{L}\backslash \{i_{\ell'}\}}\), and thus
     \begin{equation}
         \label{eq:contra1}
          \mathbf{N}_{\mathbf{D},\mathcal{M}^\star}(\ell',:)\mathbf{D}_\mathcal{L} = \beta_{\ell'}\mathbf{e}_{\ell'}^\intercal
     \end{equation}
     where \(\beta_{\ell'}\in \mathbb{F}_q\backslash \{0\}\) and \(\mathbf{e}_{\ell'}\in \mathbb{F}_q^{L\times 1}\) is an \(L\)-length vector with a \(1\) in the \(\ell'\)-th position and \(0\)-s elsewhere. 
     On the other hand, from \eqref{eq:contra0}, we have
     \begin{align}
         \mathbf{N}_{\mathbf{D},\mathcal{M}^\star}(\ell',:)\mathbf{D}_\mathcal{L} &= \left(\sum_{\ell=1,\ell\neq \ell'}^L \alpha_\ell \mathbf{N}_{\mathbf{D},\mathcal{M}^\star}(\ell,:)\right)\mathbf{D}_\mathcal{L}\notag\\
         & = \sum_{\ell=1,\ell\neq \ell'}^L \alpha_\ell \bigg(\mathbf{N}_{\mathbf{D},\mathcal{M}^\star}(\ell,:)\mathbf{D}_\mathcal{L}\bigg)\notag\\
         & = \sum_{\ell=1,\ell\neq \ell'}^L \alpha_\ell\left(\beta_\ell\mathbf{e}_\ell^\intercal\right)\label{eq:contra2}
     \end{align}
     where \(\beta_\ell \in \mathbb{F}_q\backslash \{0\}\), for every \(\ell\in [L]\backslash \{\ell'\}\). Clearly, \eqref{eq:contra1} and \eqref{eq:contra2} contradict. Therefore, \(\mathbf{N}_{\mathbf{D},\mathcal{M}^\star}\) cannot have rank less than \(L\). In other words, \(\text{rank}_q(\mathbf{N}_{\mathbf{D},\mathcal{M}^\star} )=L\), and therefore, for any given \(K,L\), and \(M\) with \(K=L+M-1\), and for any \(\mathbf{D}\in \mathbb{F}_q^{L\times K}\), the rate 
     \begin{equation*}
        R_1(K,L,M)=L
    \end{equation*}
    is achievable. This completes the proof of Lemma~\ref{lemma1}.
\end{proof}



We continue with the proof of Theorem~\ref{thm1}, by proceeding to Case 2. Before doing so, we provide a simple example of Scheme 1 -- as it pertains to Case 1 -- and demonstrate its achievable rate.

\begin{example}[Scheme 1, Case 1]
Consider the \((K=4,L=3,M=2)\) distributed linearly separable function computation problem, where the demand matrix is
    \begin{equation}
    \label{eq:ex1D}
\mathbf{D}
=
\begin{bmatrix}
    1 & 1 & 1 & 1 \\
    1 & 2 & 3 & 2 \\
    1 & 4 & 1 & 2
\end{bmatrix}.
\end{equation}
We now design the task assignment and the server transmissions based on Lemma~\ref{lemma1}. The number of servers required to employ the task assignment in Lemma~\ref{lemma1} (cf. \eqref{eq:nullspace}) is \(N=L=3\).
Note that the first three columns of \(\mathbf{D}\) are linearly independent. Thus, following the nullspace-based approach outlined in \eqref{eq:nullspace}, we get
\begin{align*}
    \mathcal{M}_1 =\{1,4\},\quad
    \mathcal{M}_2 =\{2,4\},\quad
    \mathcal{M}_3 =\{3,4\}.
\end{align*}
Note that the first server does not have access to \(W_2\) and \(W_3\). Now, we find a vector that resides in the left nullspace of the submatrix \( \mathbf{D}_{\mathcal{M}_1^c}=\mathbf{D}_{\{2,3\}}\). For instance, \([10, -3, -1 ]\) is a vector in the left nullspace of \(\mathbf{D}_{\{2,3\}}\). Then, the transmission made by the first server is
\begin{align}
    \mathbf{x}_1 &= [10, -3, -1][F_1(\mathcal{W}),F_2(\mathcal{W}),F_3(\mathcal{W})]^\intercal\notag\\
    &=6f_1(W_1)+2f_4(W_4).\label{eq:ex1x1}
\end{align}
Similarly, the vectors \([-1, 0, 1]\) and \([-2, 3, -1]\) are in the left nullspaces of \(\mathbf{D}_{\mathcal{M}_2^c} = \mathbf{D}_{\{1,3\}}\) and \( \mathbf{D}_{\mathcal{M}_3^c} = \mathbf{D}_{\{1,2\}}\), respectively. Therefore, the transmissions from server 2 and server 3 are
\begin{align}
    \mathbf{x}_2 &= [-1, 0, 1][F_1(\mathcal{W}),F_2(\mathcal{W}),F_3(\mathcal{W})]^\intercal\notag\\
    &=3f_2(W_2)+f_4(W_4)\label{eq:ex1x2}
\end{align}
and
\begin{align}
    \mathbf{x}_3 &= [-2, 3, -1][F_1(\mathcal{W}),F_2(\mathcal{W}),F_3(\mathcal{W})]^\intercal\notag\\
    &=6f_3(W_3)+2f_4(W_4) \label{eq:ex1x3}
\end{align}
respectively. Since the matrix constructed using \eqref{eq:NullMat} 
    \[
\mathbf{N}_{\mathbf{D},\mathcal{M}}
=
\begin{bmatrix}
    10 & -3 & -1  \\
    -1 & 0 & 1 \\
    -2 & 3 & -1
\end{bmatrix}
\]
is invertible, the user can decode \(F_\ell(\mathcal{W})\) for \(\ell=1,2,3\), from \(\mathbf{x}_1,\mathbf{x}_2\) and \(\mathbf{x}_3\). Therefore, from \eqref{eq:ex1x1}, \eqref{eq:ex1x2}, and \eqref{eq:ex1x3}, the demanded functions in \eqref{eq:ex1D} can be decoded as 
\begin{align*}
F_1(\mathcal{W}) &=\frac{1}{6}(\mathbf{x}_1+2\mathbf{x}_2+\mathbf{x}_3),\\ 
F_2(\mathcal{W}) &=\frac{1}{6}(\mathbf{x}_1+4\mathbf{x}_2+3\mathbf{x}_3),\\
F_3(\mathcal{W}) &=\frac{1}{6}(\mathbf{x}_1+8\mathbf{x}_2+\mathbf{x}_3)
\end{align*}
and thus the rate achieved is
\[R_1(K=4,L=3,M=2) = 3.\]
Since \(\text{rank}_q(\mathbf{D})=3\), even a centralized server with \(M=4\), having access to all the datasets, would require \(3\) transmissions to meet the user's requests. Therefore, the optimal rate of the \((K=4,L=3,M=2)\) distributed linearly separable function computation is \[R^*(4,3,2) = R_1(4,3,2)=3.\]
This completes the example. 
\end{example}

We continue with the proof, moving on to Case~2.  The nullspace-based condition in Lemma~\ref{lemma1} essentially ensures the existence of \(L\) linearly independent vectors, each with a Hamming weight (or support size) less than or equal to \(M\), in the rowspace of \(\mathbf{D}\). However, now, for \(K>L+M-1\), for any \(\mathcal{M}_n\subseteq [K]\) with \(|\mathcal{M}_n|\leq M\), the number of columns in the submatrix \(\mathbf{D}_{\mathcal{M}_n^c}\) is greater than or equal to \(K-M\), where \(K-M>L-1\). Subsequently, \(\mathbf{D}_{\mathcal{M}_n^c}\) need not have a non-trivial left nullspace in general. To combat this, we devise two techniques. In the next subsection, we present an appropriate column-wise partition of the demand matrix into submatrices (calling them sub-demand matrices), and use the technique presented in the proof of Lemma~\ref{lemma1} in each of the sub-demand matrices separately. Later in Theorem~\ref{thm2}, we present another technique to build nullspaces by augmenting the rows of the demand matrix.

\subsubsection{\textbf{Case 2} (\(K>L+M-1\))}
\label{Case2}~\\
When \(K>L+M-1\), we partition the demand matrix \(\mathbf{D}\in \mathbb{F}_q^{L\times K}\) into \(\ceil{K/K'}\) sub-demand matrices, for some positive integer \(K'\leq L+M-1\), as follows 
\begin{equation*}
    \mathbf{D} = [\mathbf{D}_1,\mathbf{D}_2,\dots, \mathbf{D}_{\ceil{K/K'}}]
\end{equation*}
where \(\mathbf{D}_\nu\in \mathbb{F}_q^{L\times K'}\) for all \(\nu\in [\ceil{K/K'}-1]\), and where \(\mathbf{D}_{\ceil{K/K'}}\in \mathbb{F}_q^{L\times (K-K'(\ceil{K/K'}-1))}\). Rewriting (\ref{eq:rep}), we have
\begin{subequations}
\label{eq:aligned}
    \begin{align}
\big[F_1(\mathcal{W}),F_2(\mathcal{W})&,\dots,F_L(\mathcal{W})\big]^\intercal = \mathbf{D}\left[f_1(W_1),f_2(W_2),\dots,f_K(W_K)\right]^\intercal\\
    & = \sum_{\nu=1}^{\ceil{K/K'}-1} \mathbf{D}_\nu\left[f_{(\nu-1)K'+1}(W_{(\nu-1)K'+1}),\dots,f_{\nu K'}(W_{\nu K'})\right]^\intercal\notag\\
    &\qquad +\mathbf{D}_{\ceil{K/K'}}\left[f_{(\ceil{K/K'}-1)K'+1}(W_{(\ceil{K/K'}-1)K'+1}),\dots,f_{K}(W_{K})\right]^\intercal\\
    &= \sum_{\nu=1}^{\ceil{K/K'}-1} \left[F_1^{(\nu)}(W_{[(\nu-1)K'+1:\nu K']}),\dots,F_L^{(\nu)}(W_{[(\nu-1)K'+1:\nu K']})\right]^\intercal\notag\\
    &\qquad +\left[F_1^{(\ceil{K/K'})}(W_{[(\ceil{K/K'}-1)K'+1:K]}),\dots,F_L^{(\ceil{K/K'})}(W_{[(\ceil{K/K'}-1)K'+1:K]})\right]^\intercal.
\end{align}
\end{subequations}
Consequently, for every \(\ell\in [L]\), we get
\begin{equation}
\label{eq:Partition}
    F_\ell(\mathcal{W}) = \sum_{\nu=1}^{\ceil{K/K'}} F_\ell^{(\nu)}(W_{[(\nu-1)K'+1: \min(\nu K',K)]})
\end{equation}
\textcolor{black}{where 
\begin{equation}
\label{eq:Fellnu}
    F_\ell^{(\nu)}(W_{[(\nu-1)K'+1: \min(\nu K',K)]}) = \mathbf{D}_\nu\left[f_{(\nu-1)K'+1}(W_{(\nu-1)K'+1}),\dots,f_{\min(\nu K',K)}(W_{\min(\nu K',K)})\right]^\intercal.
\end{equation}
For every \(\ell\in [L]\) and \(\nu\in [\ceil{K/K'}]\), the function \(F_\ell^{(\nu)}(.)\) is linearly separable over datasets \(W_{[(\nu-1)K'+1: \min(\nu K',K)]}\subseteq \mathcal{W}\).} In addition, the supports of these functions (the domain of the functions)  for different values of \(\nu\) are disjoint. Since the sub-demand matrix \(\mathbf{D}_\nu\), \(\nu \in [\ceil{K/K'}]\), satisfies the condition \(K'\leq L+M-1\), the coding scheme described in Lemma~\ref{lemma1} applies, and is employed for \(\mathbf{D}_\nu\). Thus, from Lemma~\ref{lemma1}, the user can retrieve \(F_\ell^{(\nu)}\), for every \(\ell\in [L]\), with a communication cost \(L\) using a set of \(L\) servers. 
This coding strategy is applied separately for every \(\nu \in [\ceil{K/K'}]\) using distinct sets of \(L\) servers. Consequently, using \eqref{eq:Partition}, the user can retrieve the demanded functions \(\{F_\ell(\mathcal{W}):\ell\in [L]\}\). Therefore, the rate 
\begin{equation}
\label{eq:KbyK'}
    R_1(K,L,M) = L\ceil[\Big]{\frac{K}{K'}}
\end{equation}
is achievable with \(N = L\ceil[\Big]{\frac{K}{K'}}\) servers. 
To minimize the rate expression in \eqref{eq:KbyK'}, we can choose the largest possible value of \(K'\), which is \(L+M-1\). Therefore, we get the required rate expression
\begin{equation}
\label{eq:main}
    R_1(K,L,M) = L\ceil[\Big]{\frac{K}{L+M-1}}.
\end{equation}
This completes the proof of Theorem~\ref{thm1}.\hfill \(\blacksquare\)

\textcolor{black}{\begin{rem}
    \label{rem:KbyK'}
    To make the connection with the subsequent Section~\ref{RvsN}, it is worth quickly observing here that by choosing \(K'= L+M-\gamma\), \(\gamma\in [2:L]\) and by designing an alternative task assignment technique discussed in Section~\ref{RvsN}, it will be possible to achieve the rate in \eqref{eq:KbyK'} using \(N = \left\lceil\frac{L}{\gamma}\right\rceil\ceil[\Big]{\frac{K}{K'}}\) servers. Such a chosen value of \(K'< L+M-1\) would entail a higher rate, thus bringing to the fore the tradeoff between the rate and the number of servers, which is discussed in detail in Section~\ref{RvsN}.
\end{rem}}

We next present an example pertaining to Case 2 ($K>L+M-1$).

\begin{example}[Scheme 1, Case 2]
\label{Ex:Case2}
Consider the \((K=8,L=2,M=3)\) distributed linearly separable function computation problem, and consider the demand matrix 
    \[
\mathbf{D}
=
\begin{bmatrix}
    1 & 1 & 1 & 1 & 1 & 1 & 1 & 1  \\
    0 & 1 & 2 & 3 & 4 & 5 & 6 & 7 
\end{bmatrix}.
\]
Since each server can compute \(M=|\mathcal{M}_n|=3\) subfunctions, the size of submatrix \(\mathbf{D}_{\mathcal{M}_n^c}\) is \(2\times 5\) for any \(\mathcal{M}_n\subseteq [8]\). Therefore, it is not guaranteed to have a left nullspace for \(\mathbf{D}_{\mathcal{M}_n^c}\), and subsequently, it is not possible to design a transmission vector for the servers using the technique associated with Lemma~\ref{lemma1}. To that end, we partition the demand matrix into two submatrices of equal size, \(\mathbf{D}_1=\mathbf{D}_{(:,[1:4])}\) and \(\mathbf{D}_2=\mathbf{D}_{(:,[5:8])}\), that take the form
    \[
\mathbf{D}_1
=
\begin{bmatrix}
    1 & 1 & 1 & 1  \\
    0 & 1 & 2 & 3 
\end{bmatrix},\quad
\mathbf{D}_2
=
\begin{bmatrix}
     1 & 1 & 1 & 1  \\
     4 & 5 & 6 & 7 
\end{bmatrix}.
\]
Notice that \(\mathbf{D}_1\) and \(\mathbf{D}_2\) have \(K'=K/2=4\) columns, and that \(K'=L+M-1\). Now, for each submatrix, we can design separate task assignment and transmission strategies using Lemma~\ref{lemma1}. Therefore, with \(R=L=2\) transmissions (using Lemma~\ref{lemma1}), the user can obtain the following linear combinations of the subfunctions \(f_j(W_j)\), \(j \in [4]\)
\begin{align*}
    F_1^{(1)}(W_{[1:4]}) &= f_1(W_1)+f_2(W_2)+f_3(W_3)+f_4(W_4),\\
    F_2^{(1)}(W_{[1:4]}) &= f_2(W_2)+2f_3(W_3)+3f_4(W_4).
\end{align*}
Similarly, with a different set of two servers, the following linear combinations of \(f_j(W_j)\), \(j \in \{5,6,7,8\}\) can also be obtained by the user
\begin{align*}
    F_1^{(2)}(W_{[5:8]}) &= f_5(W_5)+f_6(W_6)+f_7(W_7)+f_8(W_8),\\
    F_2^{(2)}(W_{[5:8]}) &= 4f_5(W_5)+5f_6(W_6)+6f_7(W_7)+7f_8(W_8).   
\end{align*}
Recall that the required functions are 
\begin{align*}
    F_1(\mathcal{W}) &= F_1^{(1)}(W_{[1:4]})+F_1^{(2)}(W_{[5:8]}),\\
    F_2(\mathcal{W}) &= F_2^{(1)}(W_{[1:4]})+F_2^{(2)}(W_{[5:8]}).
\end{align*}
Therefore, with \(N=4\) servers, we achieve the rate
\[R_1(K=8,L=2,M=3) = 4.\]
Now, we argue -- drawing from our converse in Section~\ref{section:covering} -- that \(R_{1}(8,2,3)=4\) is, in fact, the optimal rate for the considered parameter setting. When \(K=8\) and \(M=3\), at least 3 servers are required to store all the datasets. We now show that it is not possible to achieve a rate less than \(4\) with \(N=3\) servers. Since each dataset has to be assigned to at least one server, out of \(K=8\) datasets, only one will be assigned to two servers, while the remaining seven will each be assigned to a single server. Consequently, each server will be assigned at least two datasets that are not assigned to the remaining two servers. 
Assume that \(\{k_{1}^{(n)}, k_{2}^{(n)}\}\) are the indices of the datasets that are assigned exclusively to server \(n\), for every \(n\in [3]\). That is, \(W_{k_{1}^{(1)}}\) and \(W_{k_{2}^{(1)}}\) are assigned to the first server, but not to the second and the third server. Since the submatrix \(\mathbf{D}_{\{k_1^{(n)},k_2^{(n)}\}}\) is full rank for every choice of \(\{k_{1}^{(n)}, k_{2}^{(n)}\}\), server \(n\) has to make at least two transmissions, as no other server can make transmissions involving the subfunctions  \(f_{k_{1}^{(n)}}(W_{k_{1}^{(n)}})\) and \(f_{k_{2}^{(n)}}(W_{k_{2}^{(n )}})\). Then, the achievable rate with \(N=3\) servers must satisfy \(R\geq6\).
Therefore, the rate achieved by our scheme \(R_{1}(8,2,3)=4\) with \(N=4\) servers is optimal. Note that, if there are more than 4 servers transmitting, the achievable rate is strictly more than 4, since we do not consider subpacketization of the computed subfunction values. 
\end{example}

The converse idea discussed in Example \ref{Ex:Case2} is generalized and connected to a combinatorial structure called the multi-level covering number, which is discussed in detail in Section \ref{section:covering}.

Before proceeding with Scheme 2, we provide an additional example of Scheme 1 for Case~2. The example highlights the reduction in rate compared to a scheme having a disjoint task assignment policy.
\begin{example}[Scheme 1, Case 2]
Consider the \((K=6,L=2,M=2)\) distributed linearly separable function computation problem. First, consider the disjoint task assignment strategy \(\mathcal{M}^{(\text{disjoint})}= \{\mathcal{M}^{(\text{disjoint})}_1, \mathcal{M}^{(\text{disjoint})}_2,\mathcal{M}^{(\text{disjoint})}_3\}\), where
\begin{align*}
    \mathcal{M}^{(\text{disjoint})}_1=\{1,2\},\quad
    \mathcal{M}^{(\text{disjoint})}_2=\{3,4\},\quad
    \mathcal{M}^{(\text{disjoint})}_3=\{5,6\}.
\end{align*}
Let the demand matrix be
    \[
\mathbf{D}
=
\begin{bmatrix}
    1 & 1 & 1 & 1 & 1 &1 \\
    0 & 1 & 2 & 3 & 4 & 5 
\end{bmatrix}\in \mathbb{F}_{7}^{2\times 6}.
\]
In the case of disjoint task assignment, where each server makes two transmissions, the transmissions made by the respective servers are 
\begin{align*}
    &\mathbf{x}_{1,1}=f_1(W_1)+f_2(W_2), \quad \mathbf{x}_{1,2}=f_2(W_2),\\
    &\mathbf{x}_{2,1}=f_3(W_3)+f_4(W_4), \quad \mathbf{x}_{2,2}=2f_3(W_3)+3f_4(W_4),\\
    &\mathbf{x}_{3,1}=f_5(W_5)+f_6(W_6), \quad \mathbf{x}_{3,2}=4f_5(W_5)+5f_6(W_6).
\end{align*}
Thus, the rate achieved is
\[R_{\text{disjoint}}^*(6,2,2) = 6\] with \(N=3\) servers.\\
Corresponding to the scheme in Theorem~\ref{thm1}, we have the task assignment \(\mathcal{M}=\{\mathcal{M}_1,\mathcal{M}_2,\mathcal{M}_3,\mathcal{M}_4\}\), where
\begin{align*}
    \mathcal{M}_1 &= \{1,3\}, \quad \mathcal{M}_2 = \{2,3\},\\
    \mathcal{M}_3 &= \{4,6\}, \quad \mathcal{M}_4 = \{5,6\}.
\end{align*}
Thus, from Theorem~\ref{thm1}, we have the server transmissions 
\begin{align*}
    &\mathbf{x}_1=f_1(W_1)-f_3(W_3),\\
    &\mathbf{x}_2=f_2(W_3)+2f_3(W_3),\\
    &\mathbf{x}_3=f_4(W_4)-f_6(W_6),\\
    &\mathbf{x}_4=f_5(W_5)+2f_6(W_6).
\end{align*}
Thus, the rate achieved by our proposed scheme is
\[R_{1}(6,2,2) = 4.\]
\end{example}

We note that until now, in Scheme 1, after establishing the task assignment, the delivery follows closely in the footsteps of \cite{WSJC} in designing the nullspaces that govern transmission. The nullspaces are naturally different from \cite{WSJC} which relies on fixed task assignment (disjoint or cyclic), but the design method -- after establishing our task assignment -- is similar. In our setting though, this nullspace approach guarantees decodability for all demand matrices. 

We now proceed with Scheme 2 which includes an additional design that carefully augments the demand matrix to build non-trivial nullspaces for the sub-demand matrices.  


\subsection{Scheme 2: Alternate Scheme Better Suited for Larger \(M\)}
In Section \ref{Case2}, we devised a matrix partitioning technique to build the nullspace condition in \eqref{eq:nullspace} for a given demand matrix. Instead, in this section, we propose another technique for building a target nullspace and satisfying the nullspace condition in Lemma~\ref{lemma1}, where this now involves augmenting the demand matrix. The proposed technique is applicable for large values of \(M\), specifically when \(M\geq K/2\). Since we already have an optimal scheme for the case \(L>K-M\) (Case 1 in Scheme 1), we focus on the scenario where \(L\leq K-M\), and we have the following theorem.

\begin{thm}
    \label{thm2}
        For the $(K,L,M)$ distributed linearly separable function computation problem with \(L<\floor{K/(K-M)}\), 
        the rate
        \begin{equation}
        \label{eq:thm2eq1}
        R_2(K,L,M) = L+1
    \end{equation}
    is achievable with \(N=L+1\) servers.
    Furthermore, for any \(L\leq K-M\), the rate
    \begin{equation}
    \label{eq:L+X}
        R_2(K,L,M) = L+\ceil[\Bigg]{\frac{L}{\floor[\big]{\frac{M}{K-M}}}}
    \end{equation}
    is achievable with \(N=\floor{K/(K-M)}\), if $M\geq K/2$.
\end{thm}
\begin{proof}
    Let \(\mathbf{D}\in \mathbb{F}_q^{L\times K}\) be the demand matrix. First, we consider the case \(L<\tau\), where \(\tau=\floor{K/(K-M)}=1+\floor{M/(K-M)}\) is a positive integer. We set the number of servers \(N=L+1\leq\tau\). To design the task assignment \(\mathcal{M}=\{\mathcal{M}_1,\mathcal{M}_2,\dots,\mathcal{M}_N\}\), we first partition the set \([K]\) as
    \begin{equation*}
    \mathcal{P} = \{\mathcal{P}_1,\mathcal{P}_2,\dots,\mathcal{P}_\tau\}
    \end{equation*}
    where \(|\mathcal{P}_t|= K-M\), for every \(t\in [\tau-1]\), and \(|\mathcal{P}_\tau|= K - (\tau-1)(K-M)\). Note that 
    \begin{align*}
        |\mathcal{P}_\tau|&= K - (\tau-1)(K-M)\\
        &=K-\floor{M/(K-M)}(K-M)\\
        &\geq K-M.
    \end{align*}
    Without loss of generality, we assume that for every \(t\in [\tau-1]\) then
    \begin{equation}
    \label{eq:pt}
        \mathcal{P}_t = \{(t-1)(K-M)+1,(t-1)(K-M)+2,\dots,t(K-M)\}
    \end{equation}
     and that 
     \begin{equation}
     \label{eq:ptau}
        \mathcal{P}_\tau = \{(\tau-1)(K-M)+1,(\tau-1)(K-M)+2,\dots,K\}
    \end{equation}
which means that the task assignment set \(\mathcal{M}\) takes the form 
    \begin{equation}
        \label{eq:datascheme2}
        \mathcal{M}=\{\mathcal{P}_1^c,\mathcal{P}_2^c,\dots,\mathcal{P}_N^c\}.
    \end{equation}
    Note that \(|\mathcal{M}_n|\leq M\), for every \(n\in [N]\). 
    Now, our objective is to design the server transmissions using the nullspace approach presented in Lemma~\ref{lemma1}. However, the left nullspace of \(\mathbf{D}_{\mathcal{M}_n^c}\), for any \(n\in [N]\), does not need to have a non-trivial vector, since \(L\leq K-M\). Considering that, we augment the demand matrix by an additional row such that the augmented demand matrix, denoted by \(\tilde{\mathbf{D}}\), satisfies the required condition in \eqref{eq:nullspace}. Then
    \begin{equation*}
        \tilde{\mathbf{D}}
=
\begin{bmatrix}
    \mathbf{D}\\
     \tilde{\mathbf{d}}
\end{bmatrix} \in \mathbb{F}_{q}^{(L+1)\times K}
    \end{equation*}
    where the construction of \(\tilde{\mathbf{d}}\in \mathbb{F}_{q}^{1\times K}\) is explained in the sequel. First, we define a matrix 
    \begin{equation}
        \label{eq:Target}
        \mathbf{T} = \begin{bmatrix}
            1 & 0 &  \dots & 0 & 1\\
            0 & 1 &  \dots & 0 & 1\\
            \vdots & \vdots  & \ddots  & \vdots & 1\\
            0 & 0 &  \dots & 1 & 1\\
            0 & 0 &  \dots & 0 & 1\\
        \end{bmatrix} \in \mathbb{F}_q^{N\times N}
    \end{equation}
    which we refer to as the `target nullspace matrix'. Note that \(\text{rank}_q(\mathbf{T})=N\) irrespective of the field size \(q\). We now construct the vector \(\tilde{\mathbf{d}}\) as a concatenation of several vectors, where
    \begin{equation*}
         \tilde{\mathbf{d}} = [ \tilde{\mathbf{d}}_{1},\tilde{\mathbf{d}}_2,\dots,\tilde{\mathbf{d}}_N]
    \end{equation*}
    if \(\tau=\frac{K}{K-M}\) (when \((K-M)\) divides \(K\)) and \(L=\tau-1\). In that case, \(\tilde{\mathbf{d}}_{n}\in \mathbb{F}_q^{1\times (K-M)}\), for every \(n\in [N]\). If either \(\tau<\frac{K}{K-M}\) or \(L<\tau-1\), we have
    \begin{equation*}
         \tilde{\mathbf{d}} = [ \tilde{\mathbf{d}}_{1},\tilde{\mathbf{d}}_2,\dots,\tilde{\mathbf{d}}_N,\tilde{\mathbf{d}}_{N+1}]
    \end{equation*}
    where \(\tilde{\mathbf{d}}_{n}\in \mathbb{F}_q^{1\times (K-M)}\), for every \(n\in [N]\), and \(\tilde{\mathbf{d}}_{N+1}\in \mathbb{F}_q^{1\times K-N(K-M)}\). For every \(n\in [N-1]\), we now define
    \begin{equation}
        \label{eq:partn}
        \tilde{\mathbf{d}}_{n} = -\mathbf{D}_{\mathcal{P}_n}(n,:)
    \end{equation}
    where \(\mathbf{D}_{\mathcal{P}_n}(n,:)\) is the \(n\)-th row of the submatrix \(\mathbf{D}_{\mathcal{P}_n}\). Further, we define
    \begin{equation}
        \label{eq:partN}
        \tilde{\mathbf{d}}_{N} = \mathbf{0}_{1\times (K-M)}.
    \end{equation}
   Furthermore, \(\tilde{\mathbf{d}}_{N+1}\) can be set to any row vector of length \(K-N(K-M)\). 
Now, using Lemma~\ref{lemma1} on the augmented demand matrix \(\tilde{\mathbf{D}}\) and the task assignment \(\mathcal{M}\) in \eqref{eq:datascheme2}, we show the achievability of the rate \(R_2=N=L+1\).
  That is, we show that 
  \begin{equation*}
      \text{rank}_q(\mathbf{N}_{\tilde{\mathbf{D}},\mathcal{M}})=L+1
  \end{equation*}
  where
\begin{equation*}
    \mathbf{N}_{\tilde{\mathbf{D}},\mathcal{M}} = [\mathcal{N}(\tilde{\mathbf{D}}_{\mathcal{M}_1^c}^\intercal),\mathcal{N}(\tilde{\mathbf{D}}_{\mathcal{M}_2^c}^\intercal),\dots,\mathcal{N}(\tilde{\mathbf{D}}_{\mathcal{M}_N^c}^\intercal)]^\intercal.
\end{equation*}  
Recall that, \(\mathcal{M}_n^c=\mathcal{P}_n\), for every \(n\in [N]\). Therefore, for every \(n\in [N]\), we have
\begin{equation*}
    \tilde{\mathbf{D}}_{\mathcal{M}_n^c} = \tilde{\mathbf{D}}_{\mathcal{P}_n}.
\end{equation*}
However, in the submatrix \(\tilde{\mathbf{D}}_{\mathcal{P}_n}\), we have
\begin{equation}
\label{eq:tilde_d}
    \tilde{\mathbf{D}}_{\mathcal{P}_n}(n,:)+\tilde{\mathbf{D}}_{\mathcal{P}_n}(N,:)=\tilde{\mathbf{D}}_{\mathcal{P}_n}(n,:)+\tilde{\mathbf{d}}_n=\mathbf{0}_{1 \times (K-M)}
\end{equation}
for every \(n\in [N-1]\), where \eqref{eq:tilde_d} follows from the construction of \(\tilde{\mathbf{d}}_n\) in \eqref{eq:partn}. Therefore, for every \(n\in [N-1]\), the vector
\begin{equation*}
    \mathbf{e_n}+\mathbf{e_N}\in \text{span}(\mathcal{N}(\tilde{\mathbf{D}}_{\mathcal{M}_n^c}^\intercal))
\end{equation*}
where $\mathbf{e}_i$ is an $N-$length vector with a $1$ at the $i$-th position and $0$'s elsewhere. Notice that, we also have \(\mathbf{e_n}+\mathbf{e_N}=\mathbf{T}(n,:)\), for every \(n\in [N-1]\). Therefore, the vector \(\mathbf{T}(n,:)\), \(n\in [N-1]\), is a row in \(\mathbf{N}_{\tilde{\mathbf{D}},\mathcal{M}}\). Furthermore, in the submatrix  \(\tilde{\mathbf{D}}_{\mathcal{M}_N^c}=\tilde{\mathbf{D}}_{\mathcal{P}_N}\), we have
\(\tilde{\mathbf{D}}_{\mathcal{P}_N}(N,:)=\tilde{\mathbf{d}}_N=\mathbf{0}\). Thus, the vector
\begin{equation*}
    \mathbf{e_N}\in \text{span}(\mathcal{N}(\tilde{\mathbf{D}}_{\mathcal{M}_N^c}^\intercal))
\end{equation*}
and note that \(\mathbf{e_N}=\mathbf{T}(N,:)\). Consequently, the vector \(\mathbf{T}(N,:)\) is a row in \(\mathbf{N}_{\tilde{\mathbf{D}},\mathcal{M}}\). Therefore, we have
\begin{equation*}
    \text{rank}_q(\mathbf{N}_{\tilde{\mathbf{D}},\mathcal{M}})\geq \text{rank}_q(\mathbf{T})=N=L+1.
\end{equation*}
However, the number of columns in \(\mathbf{N}_{\tilde{\mathbf{D}},\mathcal{M}}\) is \(N\). Thus, the rank cannot exceed \(N\). Therefore,  
\begin{equation}
\label{eq:L+1}
    \text{rank}_q(\mathbf{N}_{\tilde{\mathbf{D}},\mathcal{M}})=L+1.
\end{equation}
The achievability of \(R_2(K,L,M)=L+1\) follows from \eqref{eq:L+1}.

 It remains to show the achievability of \eqref{eq:L+X} for a general \(L\), with \(N=\floor{K/(K-M)}\). Let \(\mathbf{D}\in \mathbb{F}_q^{L\times K}\) be the demand matrix. If \(L<\floor{K/(K-M)}\), the rate expression in \eqref{eq:L+X} reduces to 
\begin{equation*}
L+\ceil[\Bigg]{\frac{L}{\floor[\big]{\frac{M}{K-M}}}}\leq L+\ceil[\Bigg]{\frac{\floor[\big]{\frac{K}{K-M}}-1}{\floor[\big]{\frac{M}{K-M}}}}= L+\ceil[\Bigg]{\frac{\floor[\big]{\frac{M}{K-M}}}{\floor[\big]{\frac{M}{K-M}}}} =L+1.       
\end{equation*}
 In order to show the achievability of the expression in \eqref{eq:L+X} for \(L\geq \floor{K/(K-M)}\), we partition the demand matrix row-wise, and apply the technique presented in the first part of this proof on each of the sub-demand matrices. We have
 \begin{equation*}
     \mathbf{D} =\begin{bmatrix}
         \mathbf{D}_1\\
         \mathbf{D}_2\\
         \vdots\\
         \mathbf{D}_{\ceil{L/{L'}}}
     \end{bmatrix}
 \end{equation*}
 where \(L' = \floor{K/(K-M)}-1=\floor{M/(K-M)}\). For every \(\lambda \in [\ceil{L/{L'}}-1]\), we have the submatrices \(\mathbf{D}_\lambda\in \mathbb{F}_q^{L'\times K}\), and \(\mathbf{D}_{\ceil{L/{L'}}}\in \mathbb{F}_q^{(L-(\ceil{L/{L'}}-1)L')\times K}\). Further, 
 \begin{align*}
\big[F_1(\mathcal{W}),F_2(\mathcal{W}),\dots,F_L(\mathcal{W})\big]^\intercal = \mathbf{D}\left[f_1(W_1),f_2(W_2),\dots,f_K(W_K)\right]^\intercal
 \end{align*}
 implies that for every \(\lambda \in [\ceil{L/{L'}}-1]\) it holds that  
 \begin{align}
 \label{eq:lambda1}
     \big[F_{(\lambda-1)L'+1}(\mathcal{W}),F_{(\lambda-1)L'+2}(\mathcal{W}),\dots,F_{\lambda L'}(\mathcal{W})\big]^\intercal = \mathbf{D}_\lambda\left[f_1(W_1),f_2(W_2),\dots,f_K(W_K)\right]^\intercal
 \end{align}
and 
 \begin{equation}
 \label{eq:lambda2}
     \big[F_{(\ceil{L/{L'}}-1)L'+1}(\mathcal{W}),F_{(\ceil{L/{L'}}-1)L'+2}(\mathcal{W}),\dots,F_{L}(\mathcal{W})\big]^\intercal = \mathbf{D}_{\ceil{L/{L'}}}\left[f_1(W_1),f_2(W_2),\dots,f_K(W_K)\right]^\intercal.
 \end{equation}
 Since \(L'<\floor{K/(K-M)}\), the task assignment and transmissions based on a target nullspace matrix, presented in the first part of this proof, are applicable in each sub-demand matrix \(\mathbf{D}_\lambda\), \(\lambda \in [\ceil{L/{L'}}]\). Note that the task assignment set \(\mathcal{M}\) is independent of \(\mathbf{D}_\lambda\), and we have
 \begin{equation*}
      \mathcal{M}=\{\mathcal{P}_1^c,\mathcal{P}_2^c,\dots,\mathcal{P}_N^c\}
 \end{equation*}
 where the sets \(\mathcal{P}_n, n\in [N]\), are defined in \eqref{eq:pt} and \eqref{eq:ptau}. 
 For each \(\mathbf{D}_\lambda\), the transmission of the servers can be designed using the target nullspace matrix. The functions corresponding to each \(\mathbf{D}_\lambda\), \(\lambda \in [\ceil{L/{L'}}]\) can thus be decoded by the user (from \eqref{eq:lambda1} and \eqref{eq:lambda2}).

 We now calculate the rate achieved by this distributed computing system. Recall that there are \(N=L'+1\) servers in the system. Corresponding to every \(\mathbf{D}_\lambda\), \(\lambda\in [\ceil{L/{L'}}-1]\), each server makes a transmission, and corresponding to \(\mathbf{D}_{\ceil{L/{L'}}}\), the first \(L-(\ceil{L/{L'}}-1)L'+1\) servers make one transmission each and the remaining \(N-\left(L-(\ceil{L/{L'}}-1)L'+1\right)=L'\ceil{L/{L'}}-L\) servers do not make any transmission. Thus, the rate is
 \begin{align*}
     R_2(K,L,M) &= N\left(\left\lceil{\frac{L}{L'}}\right\rceil-1\right)+ L-(\ceil{L/{L'}}-1)L'+1\\
     R_2(K,L,M) &= (L'+1)\left(\left\lceil{\frac{L}{L'}}\right\rceil-1\right)+ L-(\ceil{L/{L'}}-1)L'+1\\
     &=L+\left\lceil{\frac{L}{L'}}\right\rceil\\
     &=L+\left\lceil{\frac{L}{\left\lfloor\frac{M}{K-M}\right\rfloor}}\right\rceil.
 \end{align*}
 This completes the proof of Theorem~\ref{thm2}.
\end{proof}
\begin{rem}
    \label{rem:Dindependent}
     Unlike Scheme 1, the task assignment \(\mathcal{M}\) in Scheme 2 is independent of the demand matrix. However, if \(L<\floor{K/(K-M)}-1\), only \(L+1\) servers need to perform the subfunction computations and the corresponding transmissions; not all servers need to participate in the computation phase.
\end{rem}

\begin{rem}
    \label{rem:TargetNullSpace}
    The target nullspace matrix \(\mathbf{T}\) need not be the one defined in \eqref{eq:Target}.
    It can be any \(N\times N\) matrix over \(\mathbb{F}_q\) that satisfies the following two properties.
    \begin{enumerate}
        \item \(\text{rank}_q(\mathbf{T})=N\).
        \item The entries of the \(N\)-th column of \(\mathbf{T}\) are all non-zero.
    \end{enumerate}
    The first property ensures the decodability of the demanded functions, and the second property is required to construct the \(L+1\)-th row. The second property also ensures the existence of a nullspace for the corresponding submatrix of the augmented demand matrix \(\tilde{\mathbf{D}}\). Having zero as the \(N\)-th entry of a row implies the existence of a non-trivial nullspace for the corresponding submatrix of the original demand matrix \(\mathbf{D}\), which is not guaranteed.   
\end{rem}
\begin{example}[Scheme 2]
\label{ex:scheme2} 
    Consider the \((K=6,L=2,M=4)\) distributed linearly separable function computation problem with a demand matrix
    \[
\mathbf{D}
=
\begin{bmatrix}
    1 & 1 & 1 & 1 & 1 &1 \\
    0 & 1 & 2 & 3 & 4 & 5 
\end{bmatrix}
\]
and note that \(L<K/(K-M)=3\).
From the scheme described in the proof of Theorem~\ref{thm2}, we have \(N=L+1=3\). The task assignment on the three servers is
\begin{align*}
    \mathcal{M}_1 &= \{1,2,3,4\},\\
    \mathcal{M}_2 &= \{1,2,5,6\},\\
    \mathcal{M}_3 &= \{3,4,5,6\}.
\end{align*}
We choose the matrix
 \[
\mathbf{T}
=
\begin{bmatrix}
    1 & 0 & -1 \\
    0 & 1 & -1 \\
    0 & 0 & 1
\end{bmatrix}
\]
as the target nullspace matrix. Then the augmented demand matrix \(\tilde{\mathbf{D}}\) is constructed in such a way that the vector \(\mathbf{T}(1,:) = [1, 0, -1]\) falls in the left nullspace of the submatrix \(\tilde{\mathbf{D}}_{\{5,6\}}\).  
Similarly, the vectors \([0, 1, -1]\) and \([0, 0, 1]\) should be in the left nullspaces of \(\tilde{\mathbf{D}}_{\{3,4\}}\) and \(\tilde{\mathbf{D}}_{\{1,2\}}\), respectively.
Therefore, we get  
    \[
\tilde{\mathbf{D}}
=
\begin{bmatrix}
    1 & 1 & 1 & 1 & 1 & 1 \\
    0 & 1 & 2 & 3 & 4 & 5 \\
    0 & 0 & 2 & 3 & 1 & 1
\end{bmatrix}.
\]
The following expressions denote the linear combinations of the computed subfunctions sent by the three servers, respectively:
\begin{align*}
     &\quad \mathbf{x}_1 = f_1(W_1)+f_2(W_2)-f_3(W_3)-2f_4(W_4)\\
    &\quad \mathbf{x}_2 = f_2(W_2)+3f_5(W_5)+4f_6(W_6)\\
    &\quad \mathbf{x}_3 = 2f_3(W_3)+3f_4(W_4)+f_5(W_5)+f_6(W_6).
\end{align*}
Note that the user can decode the required functions, since
\begin{align*}
    F_1(\mathcal{W})=\mathbf{x}_1+\mathbf{x}_3,\quad
    F_2(\mathcal{W})=\mathbf{x}_2+\mathbf{x}_3.
\end{align*}

\end{example}
\begin{rem}
    \label{rem:Ex3}
    In Example~\ref{ex:scheme2}, instead of adding the particular row that we have added, it is, in fact, possible to add any random third row and achieve the rate \(R_1(6,2,4)=3\) using Scheme 1. This holds because when \(L=3\) and \(M=4\), we have \(K=L+M-1=6\), and in this case, Lemma~\ref{lemma1} provides the achievability of \(R_1(6,2,4)=3\). However, it is not the case in general. For instance, when \(K=9\), \(M=6\), and \(L=2\), the achievable rate corresponding to Scheme 2 is \(R_2=3\), while Scheme 1 requires 2 additional random rows to fall into the regime \(K=L+M-1\), and the corresponding rate would be \(R_1=4\). 
\end{rem}
The achievability result in \eqref{eq:L+X} is obtained by partitioning the demand matrix, row-wise, into sub-demand matrices and applying the achievability result in \eqref{eq:thm2eq1} in each sub-demand matrix. 
These results in Theorem~\ref{thm2} can also be applied on sub-demand matrices obtained by dividing a demand matrix column-wise, as in the proof of Theorem~\ref{thm1}.
In certain instances, this could result in a lower rate than the scheme in Theorem~\ref{thm1}. One such example is presented below.
\begin{example}[Scheme 2 with column partitioning]
   Let \(K=40, M =15\), and \(L=3\). For this setting, the scheme in Theorem~\ref{thm1} achieves the rate \[R_1(40,3,15) = L\left\lceil\frac{K}{L+M-1}\right\rceil=9\] with \(N=9\) servers. However, partitioning \(\mathbf{D}\) into two equal-sized (\(3\times 20\)) submatrices \(\mathbf{D}_1\) and \(\mathbf{D}_2\) and applying Scheme 2 separately on these submatrices will result in a rate \[R_{\text{ach}}(40,3,15)=2(L+1)=8\] with \(N=8\) servers.
\end{example}

\section{Converses and Order-Optimality}
\label{sec:converse}
In this section, we first derive an information-theoretic lower bound on the rate of a general \((K,L,M)\) distributed linearly separable function computation problem. The lower bound assumes that the system can utilize any number of servers, and that each server transmits linear combinations of its computed subfunction outputs, without any subpacketization. With the converse in place, we will then show that our proposed scheme in Theorem~\ref{thm1} is within a constant gap (gap is upper bounded by~2 when a divisibility condition is met, and the gap is upper bounded by~3, in general) from the derived information-theoretic lower bound, when the field size $q\geq eK/M$. 
\textcolor{black}{We will subsequently derive an additional converse using covering designs, and we will show that for a specific case, the new converse bound matches the achievable rate in Theorem~\ref{thm1}.}

\textcolor{black}{\subsection{Information-Theoretic Converse}
\label{section:ITconverse}}
\begin{thm}
For the $(K,L,M)$ distributed linearly separable function computing problem, the optimal rate $R^{*}(K,L,M)$ satisfies
\begin{equation}
R^{*}(K,L,M) \geq \max \Bigl( L, \frac{LK}{L+M+\log_q \binom{K}{M}}\Bigr).
\label{eq:converse}
\end{equation}
\label{th:conv}
\end{thm}

\begin{proof}
Since the user wants $L$ independent linear combinations of $f_j(W_j)$'s, we get 
\begin{equation}
 R^{*}(K,L,M) \geq L.
 \label{eq:conv1}
\end{equation}
To show $R^{*}(K,L,M) \geq \frac{LK}{L+M+\log_q \binom{K}{M}}$,  we resort to the equivalent matrix factorization formulation
of the $(K,L,M)$ distributed computing problem. Recall that given a matrix $\mathbf{D} \in \mathbb{F}_q^{L \times K}$ of $\text{rank}_q(\mathbf{D})=L$, we aim to jointly design two matrices $\mathbf{C} \in \mathbb{F}_q^{L \times R}$ and $\mathbf{A} \in \mathbb{F}_q^{R \times K}$ such that $\mathbf{D}=\mathbf{C}\mathbf{A}$, with the following constraints $||\mathbf{A}(r,:) ||_0 \leq M, \forall r \in [R]$, and $\text{rank}_q(\mathbf{A})=R \geq L$. Assume that there exists such a decomposition $\mathbf{D}=\mathbf{C}\mathbf{A}$. Then, the conditional joint entropy of the matrices $\mathbf{A}$ and $\mathbf{C}$ given $\mathbf{D}$ can be expressed as follows
\begin{subequations}
\begin{align}
    H(\mathbf{A},\mathbf{C}|\mathbf{D}) & = 
    H(\mathbf{A}|\mathbf{D})+H(\mathbf{C}|\mathbf{D},\mathbf{A})\label{eq:joint1}\\ 
    & = H(\mathbf{A}|\mathbf{D}). \label{eq:joint2}
\end{align}
\label{eq:joint}
\end{subequations}
Note that all entropies mentioned in this section are computed to the base $q$. The transition from \eqref{eq:joint1} to \eqref{eq:joint2} is because the matrix $\mathbf{C}$ is deterministic when the matrices $\mathbf{D}$ and $\mathbf{A}$ are known, i.e., given two matrices $\mathbf{D} \in \mathbb{F}_q^{L \times K}$ and $\mathbf{A} \in \mathbb{F}_q^{R \times K}$ with $\text{rank}_q(\mathbf{D})=L$, $\text{rank}_q(\mathbf{A})=R \geq L$, and $\text{Rowspace}(\mathbf{D}) \subseteq \text{Rowspace}(\mathbf{A}) $, then there exists a  unique $\mathbf{C}$ such that $\mathbf{D}=\mathbf{C}\mathbf{A}$. Therefore, we get
$H(\mathbf{C}|\mathbf{D}, \mathbf{A})=0$. Equation \eqref{eq:joint} can be equivalently written as 
\begin{equation}
 H(\mathbf{A},\mathbf{C}|\mathbf{D}) = H(\mathbf{C}|\mathbf{D})+H(\mathbf{A}|\mathbf{D},\mathbf{C}).
 \label{eq:jointeq}
\end{equation}
Since $H(\mathbf{A}|\mathbf{D},\mathbf{C}) \geq 0$, from \eqref{eq:joint2} and \eqref{eq:jointeq}, we have
\begin{equation}
H(\mathbf{C}|\mathbf{D}) \leq H(\mathbf{A}|\mathbf{D}).
\end{equation}
Equivalently,
\begin{equation}
H(\mathbf{C},\mathbf{D})  \leq H(\mathbf{A},\mathbf{D}).
\label{eq:jointineq}
\end{equation}
We know that $H(\mathbf{A},\mathbf{D})= H(\mathbf{A}) + H(\mathbf{D} | \mathbf{A})$, therefore, an upper bound for $H(\mathbf{A},\mathbf{D})$ can be easily computed by upper bounding the terms $H(\mathbf{A})$ and $H(\mathbf{D} | \mathbf{A})$.

Let us first compute $H(\mathbf{A})$. Recall that the rows of the matrix $\mathbf{A} \in \mathbb{F}_q^{R \times K}$ are $M-$sparse. Hence, the total number of possible configurations of $\mathbf{A}(r,:)$, $r \in [R]$, is at most $\binom{K}{M}q^M$, where $\binom{K}{M}$ accounts for the choice of positions of the $M$ non-zero entries, and $q^M$ corresponds to the maximum choices for the values in those $M$ positions. Therefore, we get
\begin{align}
  H(\mathbf{A}) & \leq \log_q \left (\binom{K}{M}q^M \right)^R \nonumber \\
  & = R \Bigl( \log_q \binom{K}{M} + M \Bigr).
  \label{eq:entropyA}
\end{align}
Now, consider $H(\mathbf{D} | \mathbf{A})$, where $\mathbf{D} \in \mathbb{F}_q^{L \times K}$ and $\mathbf{A} \in \mathbb{F}_q^{R \times K}$, $R \geq L$. From the relation $\mathbf{D}=\mathbf{C}\mathbf{A}$, we know that the rows of $\mathbf{D}$ form an $L-$dimensional vector space in an $R-$dimensional space spanned by the rows of $\mathbf{A}$. 
Hence, to upper bound $H(\mathbf{D} |\mathbf{A})$, we need to compute the number of ways in which $L$ linearly independent vectors can be chosen from a given $R-$dimensional space, which is precisely $(q^R-1)(q^R-q)\times\cdots\times(q^R-q^{R-L+1})$. Therefore, 
\begin{subequations}
\begin{align}
H(\mathbf{D} | \mathbf{A}) & \leq \log_q\Bigl( (q^R-1)(q^R-q)\cdots(q^R-q^{R-L+1})\Bigr)\\
& < \log_q (q^R-1)^L \\
& < RL. \label{eq:condent}
\end{align}
\end{subequations}
Thus, from \eqref{eq:entropyA} and \eqref{eq:condent}, we have
\begin{equation}
 H (\mathbf{A}, \mathbf{D}) \leq R \Bigl( \log_q \binom{K}{M} + M  +L \Bigr).
 \label{eq:jointad}
\end{equation}

Now, consider $H(\mathbf{C},\mathbf{D})$. We know that 
\begin{subequations}
\begin{align}
   H(\mathbf{C},\mathbf{D}) & \geq \max \bigl(H(\mathbf{C}), H(\mathbf{D})\bigr) \label{eq:jointcd1}\\
   & = \max(LR, LK)\label{eq:jointcd2}\\
   & = LK \label{eq:jointcd3}
\end{align}
\label{eq:jointcd}
\end{subequations}
where the transition from \eqref{eq:jointcd1} to \eqref{eq:jointcd2} is because $H(\mathbf{C}) \leq LR$ and $H(\mathbf{D}) \leq LK$ (the above inequalities are met with equality if $\mathbf{C}$ and $\mathbf{D}$ are uniform i.i.d. over $\mathbb{F}_q$). Since $L \leq R \leq K$, we get \eqref{eq:jointcd3}. Thus, using \eqref{eq:jointcd3} and \eqref{eq:jointad}, equation~\eqref{eq:jointineq} becomes
\begin{equation}
 LK \leq  R \Bigl( \log_q \binom{K}{M} + M  +L \Bigr).
 \label{eq:lbineq}
\end{equation}
Upon rearranging \eqref{eq:lbineq}, we get 
\begin{equation*}
R \geq \frac{LK}{L+M+\log_q \binom{K}{M}}
\end{equation*}
which implies
\begin{equation}
R^{*}(K,L,M) \geq \frac{LK}{L+M+\log_q \binom{K}{M}}.
\label{eq:lb2ineq}
\end{equation}
Thus, by combining \eqref{eq:conv1} and \eqref{eq:lb2ineq}, we get
\begin{equation*}
  R^{*}(K,L,M) \geq \max\Bigl(L, \frac{LK}{L+M+\log_q \binom{K}{M}}\Bigr).
\end{equation*}
This completes the proof of Theorem~\ref{th:conv}.
\end{proof}
We now characterize the gap to optimal of the performance of Scheme 1 using the converse in \eqref{eq:converse}.
\begin{thm}
\label{thm:gap}
 For the $(K,L,M)$ distributed linearly separable function computation problem, the achievable rate in Theorem~\ref{thm1} satisfies the following relation when $(L+M-1)|K$.
\begin{equation}
 \frac{R_1(K,L,M)}{R^{*}(K,L,M)} \leq 1+\frac{\log_q\binom{K}{M}+1}{L+M-1}.
 \label{eq:gap1}
\end{equation}
Furthermore, \textcolor{black}{if $q\geq \frac{eK}{M}$,} 
\eqref{eq:gap1} reduces to 
 \begin{equation}
1 \leq \frac{R_1(K,L,M)}{R^{*}(K,L,M)} \leq 2.
\label{eq:2gap}
 \end{equation}
When the divisibility condition $(L+M-1)|K$ does not hold, we have
 \begin{equation}
\frac{R_1(K,L,M)}{R^{*}(K,L,M)} \leq 3.
\label{eq:4gap}
 \end{equation}
\end{thm}
\begin{proof}
We let $R_0(K,L,M) = \frac{LK}{L+M-1}$.  From \eqref{eq:lb2ineq}, we have $R^{*}(K,L,M) \geq  \frac{LK}{L+M+\log_q \binom{K}{M}}$. Therefore,
\begin{equation}
    \frac{R_0(K,L,M)}{R^{*}(K,L,M)} \leq \frac{L+M+\log_q \binom{K}{M}}{L+M-1} = 1+\frac{\log_q \binom{K}{M}+1}{L+M-1}.
    \label{eq:ineq}
 \end{equation}
Since \textcolor{black}{$\binom{K}{M} \leq \big(\frac{eK}{M}\big)^M$ for all $1\leq M\leq K$},  
equation~\eqref{eq:ineq} can be written as
 \begin{subequations}
 \begin{align}
  \frac{R_0(K,L,M)}{R^{*}(K,L,M)} & \leq 1 + \frac{M\log_q \textcolor{black}{\big(\frac{eK}{M}\big)}
  +1 }{L+M-1} \label{eq:opgap2} \\
 & \leq 1 + \frac{M +1 }{L+M-1} \label{eq:opgap3}\\
 & \leq 2 \label{eq:opgap4}.
 \end{align}
 \label{eq:opgap}
\end{subequations}
The transition from  \eqref{eq:opgap2} to \eqref{eq:opgap3} is  because  \textcolor{black}{$\log_q\big(\frac{eK}{M}\big)\leq 1$ 
for $q \geq \frac{eK}{M}$},\footnote{By applying the approximation $\binom{K}{M} \leq K^M$, we obtain the condition $q \geq K$ on the field size $q$ to ensure the gap results in Theorem \ref{thm:gap}. In the case $M = 2$, this condition extends the applicability of the results compared to the requirement $q \geq eK/M$.} 
and \eqref{eq:opgap4}
follows from the fact that $\frac{M+1}{L+M-1} \leq 1$, if $L \geq 2$. Note that, from Theorem~\ref{thm1}, we have $R_1(K,L,M) = R_0(K,L,M)  = \frac{LK}{L+M-1}$ when $(L+M-1) \mid K$. Therefore, \textcolor{black}{from \eqref{eq:opgap}} we obtain
\begin{equation}
  R_1(K,L,M) \leq 2R^{*}(K,L,M).
\end{equation}
 In general, we have  
\begin{equation}
\begin {aligned}
   R_1(K,L,M) & \leq L\left(\frac{K}{L+M-1}+1\right) = \frac{LK}{L+M-1} +L \\
   & \leq 2R^{*}(K,L,M) + L .
   \label{eq:optgap}
\end{aligned}
\end{equation}
 Equation \eqref{eq:optgap} follows from 
 \eqref{eq:ineq} and \eqref{eq:opgap}. Since $R^{*}(K,L,M) \geq L$, equation \eqref{eq:optgap} becomes
 \begin{equation*}
    R_1(K,L,M) \leq 3R^{*}(K,L,M).
 \end{equation*}
 This completes the proof of Theorem~\ref{thm:gap}.
\end{proof}

We now derive another lower bound on the rate of a general \( (K,L,M)\) distributed linearly separable function computation problem. This lower bound establishes a connection between the distributed computing problem and covering designs. Using the derived lower bound, we show exact optimality results for certain parameter settings.

\subsection{Converse Derived from Multi-Level Covering Designs}
\label{section:covering}
We proceed to introduce a new combinatorial design that will help us build our converse. 
\begin{defn}[Multi-Level Covering Designs]
\label{defn:gcd}
    Given positive integers $v, k,$ and $m$ where $v\geq k $ and $v \geq m $, we define a $(v, k, m)$ multi-level covering design to be a pair $(X,\mathcal{B})$ where $X$ is a set of $v$ elements (called points) and $\mathcal{B}=\{\mathcal{B}_1,\mathcal{B}_2,\dots,\mathcal{B}_\Omega\}$ a multiset of subsets of $X$ of size at most $k$ (called blocks) such that, for all \(m'\leq m\), every $m'$-subset of $X$ intersects at least $m'$ blocks of $\mathcal{B}$. In other words, $(X,\mathcal{B})$ is a multi-level covering design if the condition
    \begin{equation}
    \label{eq:GCD}
    |\left\{\omega \in [\Omega]:\mathcal{T}\cap \mathcal{B}_\omega\neq \varnothing\right\}|\geq m'
    \end{equation}
    holds for every \(\mathcal{T}\subseteq X\) such that \(|\mathcal{T}|=m'\leq m\). 
    A $(v, k, m)$ multi-level covering design $(X,\mathcal{B})$ is said to be optimal if:
\begin{equation}
    |\mathcal{B}| = \min\{|\mathcal{A}| : \text{ there exists a } (v, k, m) \text{ multi-level covering design } (X, \mathcal{A})\}.
\end{equation}
In this case, the cardinality of $\mathcal{B}$ is called the multi-level covering number and is denoted by $C(v, k, m)$. 
\end{defn}

We now proceed with our new converse. 
\begin{thm}
    \label{thmcovering}
        For the $(K,L,M)$ distributed linearly separable function computation problem with \(q> K\), we have 
    \begin{equation}
        R^*(K,L,M) \geq C(K, M, L) \label{combinatorialConverse}
    \end{equation}
    where $C(K, M, L)$ is the multi-level covering number with parameters \(v=K,k=M\), and \(m=L\).
\end{thm}

Before proceeding with the proof, a small remark is in order. 
\begin{rem}
The closest design to our own, is the well-known $t-(v,k,m,\lambda)$ general covering design (cf.~\cite{Fed}), after setting $t=1$ and $\lambda = m$.  The observant reader may notice though that in the classical definition~\cite{Fed}, the intersection condition is imposed only on 
$m$-subsets of $X$, with a fixed requirement on the level of intersection. In contrast, our definition requires the 
condition to hold for all $m' \leq m$, with the intersection threshold changing with $m'$. Consequently, 
the multi-level covering number $C(v,k,m)$ in our setting is never smaller than the classical 
general covering number $C_{\mathrm{GCD}}(v,k,m)$.\footnote{Indeed, every $(v,k,m)$ 
multi-level covering design is also a $(v,k,m)$ general covering design, since 
the case $m'=m$ is included. The reverse need not hold, hence 
$C(v,k,m) \geq C_{\mathrm{GCD}}(v,k,m)$.}  Hence we can conclude that the bound \eqref{combinatorialConverse} 
remains valid under either definition, since our multi-level covering number is always larger. The new definition can provide for a potentially tighter converse.
\end{rem}
We proceed with the proof.
\begin{proof}
    Let \(\mathbf{D}\in \mathbb{F}_q^{L\times K}\) be the demand matrix. Since we are interested in the worst-case communication cost, we consider \(\mathbf{D}\) to be a matrix with the property that every submatrix of \(\mathbf{D}\) of size \(L\times L\) is invertible. We know from the literature on MDS codes, that the existence of such a matrix for a given \(K\) and any \(L\leq K\) requires the field size \(q> K\) \cite{MSl,Ful}. Let \(R\) be the total number of transmissions made by the servers. Note that \(R\) need not be the number of servers, since for \(r\neq r'\), it is possible to have \(\mathcal{M}_r=\mathcal{M}_{r'}\). For every \(r\in [R]\), we have the transmission\footnote{With a slight abuse of notation, we denote the \(R\) transmissions received at the user with \(x_1,x_2,\dots x_R\) instead of the notation \(x_{n,r_n}\) used in \eqref{eq:xnr}. This mapping can be captured by defining a bijection from the set of possible \((n,r_n)\) pairs to \([R]\). } 
    \begin{equation}
        \label{eq:MatA}
        {x}_r = \sum_{   {\substack{j\in \mathcal{M}_r\\ |\mathcal{M}_r|\leq M}}}\alpha_{r,j} f_j(W_j)  
    \end{equation}
    for \(\alpha_{r,j}\in \mathbb{F}_q\), where \(\mathcal{M}_r\subseteq [K], \ |\mathcal{M}_r|\leq M\) is the support of the transmission vector \({x}_r\), \(r\in [R]\).\footnote{Note that the set $\{\mathcal{M}_1,\mathcal{M}_2,\dots,\mathcal{M}_R\}$ differs slightly from the task assignment set used previously.}  The user linearly combines the transmissions \({x}_r\), \(r\in [R]\), to obtain the demanded functions. 
    Recalling the decomposition of \(\mathbf{D}\) into transmission and decoding matrices in \eqref{eq:trans1} and \eqref{eq:dec}, we write
    \begin{align}
        \label{eq:D=CA}
        \begin{bmatrix}
            F_1(\mathcal{W})\\F_2(\mathcal{W})\\ \vdots\\F_L(\mathcal{W})
        \end{bmatrix} &= \underbrace{\begin{bmatrix}
            d_{1,1} & d_{1,2} &\dots &d_{1,K}\\
            d_{2,1} & d_{2,2} &\dots &d_{2,K}\\
            \vdots & \vdots &\ddots &\vdots \\
            d_{L,1} & d_{L,2} &\dots &d_{L,K}
        \end{bmatrix}}_{\mathbf{D}}\begin{bmatrix}
            f_1(W_1)\\f_2(W_2)\\ \vdots\\f_K(W_K)
        \end{bmatrix}\notag\\&= \underbrace{\begin{bmatrix}
            c_{1,1} & c_{1,2} &\dots &c_{1,R}\\
            c_{2,1} & c_{2,2} &\dots &c_{2,R}\\
            \vdots & \vdots &\ddots &\vdots \\
            c_{L,1} & c_{L,2} &\dots &c_{L,R}
        \end{bmatrix}}_{\mathbf{C}}\underbrace{\begin{bmatrix}
            \alpha_{1,1} & \alpha_{1,2} &\dots &\alpha_{1,K}\\
            \alpha_{2,1} & \alpha_{2,2} &\dots &\alpha_{2,K}\\
            \vdots & \vdots &\ddots &\vdots \\
            \alpha_{R,1} & \alpha_{R,2} &\dots &\alpha_{R,K}
        \end{bmatrix}}_{\mathbf{A}}\begin{bmatrix}
            f_1(W_1)\\f_2(W_2)\\ \vdots\\f_K(W_K)
        \end{bmatrix}
    \end{align}
    to get
    \begin{equation}
        \label{eq:D=CA2}
        \mathbf{D}=\mathbf{C}\mathbf{A}.
    \end{equation}
    Note that the positions of nonzero entries in the \(r\)-th row of \(\mathbf{A}\) are determined by \(\mathcal{M}_r\). Consider a matrix \(\mathbf{D}_{\mathcal{L}}\) formed by choosing the columns of \(\mathbf{D}\) indexed with the elements of the set \(\mathcal{L}\subseteq [K]\), where \(|\mathcal{L}|=L'\leq L\). From the matrix decomposition in \eqref{eq:D=CA2}, we have
    \begin{equation}
        \label{eq:DL}
        \mathbf{D}_{\mathcal{L}} = \mathbf{C}\mathbf{A}_\mathcal{L}
    \end{equation}
    where \(\mathbf{A}_{\mathcal{L}}\) is formed by choosing the columns of \(\mathbf{A}\) indexed with the elements of the set \(\mathcal{L}\).
    From our earlier assumption on the invertibility of the submatrices of \(\mathbf{D}\), we have \(\text{rank}_q(\mathbf{D}_{\mathcal{L}})=L'\) for any \(\mathcal{L}\subseteq [K]\) with \(|\mathcal{L}|=L'\leq L\). Therefore, we require \(\text{rank}_q(\mathbf{A}_\mathcal{L})=L'\). Further relaxing the requirement, we need a minimum of \(L'\) rows in \(\mathbf{A}_\mathcal{L}\) with at least a non-zero entry. In other words, we need
    \begin{equation}
        \label{eq:condition}
        |\{r:\mathcal{L}\cap \mathcal{M}_r\neq \varnothing\}|\geq L'.
    \end{equation}
    This condition must be true for every \(\mathcal{L}\subseteq [K]\) with \(|\mathcal{L}|\leq L\). Comparing \eqref{eq:condition} with \eqref{eq:GCD}, we get that in order to have a matrix decomposition as in \eqref{eq:D=CA}, the pair \(\left([K],\mathcal{M}\right)\) requires to be a \((K,M,L)\) multi-level covering design, where \(\mathcal{M}=\{\mathcal{M}_1,\mathcal{M}_2,\dots,\mathcal{M}_R\}\). However, the number of blocks in the design
    \begin{equation*}
        |\mathcal{M}|=R\geq C(K,M,L)
    \end{equation*}
    where 
    \begin{equation*}
        C(K,M,L)= \min\{|\mathcal{M}| : \text{ there exists a } (K, M, L) \text{ multi-level covering design } ([K], \mathcal{M})\}
    \end{equation*}
    is the multi-level covering number. Therefore, the rate incurred by any scheme is lower bounded by 
    \begin{equation*}
        R^*(K,L,M) \geq C(K,M,L)
    \end{equation*}
    if \(q> K\). This completes the proof of Theorem~\ref{thmcovering}.
\end{proof}

\textcolor{black}{\begin{rem}
    \label{rem:MDS}
 The proof of Theorem~\ref{thmcovering} does not strictly require the condition \(q > K\). Rather, for given \(K,L,\) and \(q\), it requires the existence of at least one matrix \(\mathbf{D} \in \mathbb{F}_q^{L \times K}\) such that every \(L \times L\) submatrix of \(\mathbf{D}\) has full rank. For instance, when \(K = L + 1\), such a matrix exists even over \(\mathbb{F}_2\). As one may notice, this is the same property attributed to the generator matrix of a maximum distance separable (MDS) code \cite{MSl}. The assumption \(q > K\) is often imposed because it guarantees the existence of simple explicit constructions, such as Vandermonde matrices, which possess the required full-rank property \cite{HoJ,Ful}.
\end{rem}
\begin{rem}
    \label{rem:generalcovering}
 A closer inspection of the proof of Theorem~\ref{thmcovering} shows that any achievable scheme -- i.e., any scheme that can successfully treat all full-rank demand matrices -- for the \((K,L,M)\) distributed linearly separable function computation problem (when \(q > K\)) must rely on an underlying multi-level covering design structure. To be more specific, such a scheme can exist only if a $(K,M,L)-$multi-level covering design exists. If such a design did not exist, no scheme would be able to resolve the problem for any demand matrix \(\mathbf{D}\) for which every \(L\times L\) submatrix is invertible\footnote{It is worth recalling that when considering the real (or the complex) number field, such matrices appear with probability 1.}. 
 To elaborate on this point, we note that the task assignments used in Theorem~\ref{thm1} and Theorem~\ref{thm2} naturally give rise to such designs:\\
    -- The task assignment \(\mathcal{M}^\star\) in Theorem~\ref{thm1}, given in \eqref{eq:dataset1}, forms the block set of a \((K,M,L)\) multi-level covering design \(([K],\mathcal{M}^\star)\), where 
\(
    |\mathcal{M}^\star| = L  \Big\lceil \frac{K}{L+M-1} \Big\rceil
\).\footnote{The column partition can produce submatrices of non-uniform size, which may result 
in some blocks having fewer than \(M\) elements. Nevertheless, condition \eqref{eq:GCD} 
is satisfied for every subset of \([K]\) of size at most \(L\).} 
Note that the number of blocks in this design corresponds to the achievable rate of the scheme.\\
    -- Unlike Theorem~\ref{thm1}, the multi-level covering design arising from Theorem~\ref{thm2} may contain repeated blocks. Consider the task assignment \(\mathcal{M} = \{\mathcal{M}_1, \mathcal{M}_2, \dots, \mathcal{M}_N\}\) in Theorem~\ref{thm2}, given in \eqref{eq:datascheme2}. The multiplicity of each block \(\mathcal{M}_n\), for \(n \in [N]\), is equal to the number of transmissions made by server \(n\). Thus, the corresponding \((K,M,L)\) multi-level covering design is \(([K], \mathcal{M}^\star)\), where \(\mathcal{M}^\star\) is the multiset obtained from \(\mathcal{M}\) by allowing repetitions, and \(|\mathcal{M}^\star| = L + \Big\lceil \frac{L}{\lfloor M / (K-M) \rfloor} \Big\rceil\).
\end{rem}}

It can be readily verified that \(C_{\mathrm{GCD}}(K,M=1,L)=K\) and \(C_{\mathrm{GCD}}(K,M,L=1)=\ceil{K/M}\). Note that, for these cases, the multi-level covering number coincides with the general covering number. i.e., \(C(K,M=1,L)=K\) and \(C(K,M,L=1)=\ceil{K/M}\). For any other values of \(L\) and \(M\), the general covering number is not known. Finding the general covering number is a hard problem, in general. Different variants of the covering design problem are studied in the literature \cite{Fed,CMR,CoD,Fed1}. For most of those problems, closed-form solutions are not available, and existing solutions \textcolor{black}{primarily rely on} probabilistic algorithms and extensive computer searches. However, using combinatorial arguments, we find a lower bound on the multi-level covering number for the specific case \(L=2\), and interestingly, the lower bound matches the rate achieved by our scheme in Theorem~\ref{thm1} when \((M+1)|K\). 
\begin{thm}
    \label{thm4}
    For the $(K,L=2,M)$ distributed linearly separable function computation problem, we have
    \begin{equation}
        R^*(K,L=2,M) = \frac{2K}{M+1}.
    \end{equation}
    if $(M+1)|K$.
\end{thm}
\begin{proof}
    When \((M+1)|K\), the achievability of the rate \(R_1(K,L=2,M) = \frac{2K}{M+1}\) follows from Theorem~\ref{thm1}. In order to show the converse, we use Theorem~\ref{thmcovering}, and prove that \(C(K,M,2)= \frac{2K}{M+1}\).
    Let \(([K],\mathcal{B})\) be a \((K,M,L)\) multi-level covering design. Let \(\mathcal{B}=\{\mathcal{B}_1,\mathcal{B}_2,\dots,\mathcal{B}_\Omega\}\) where \(\Omega\) is the number of blocks in the design, and for every \(\omega\in [\Omega]\), \(|\mathcal{B}_\omega|\leq M\). For every \(k\in [K]\), we now define
    \begin{equation}
        \label{eq:B_k}
        \mathcal{B}^{(k)} = \{\omega:\mathcal{B}_\omega\ni k\} 
    \end{equation}
    where \(\mathcal{B}^{(k)}\) is the set of indices of the blocks that contain the point \(k\). Further, for every \(\omega\in [\Omega]\), we define
    \begin{equation}
    \label{eq:K_q}
         \mathcal{K}_\omega = \{k:|\mathcal{B}^{(k)}|=\omega\}
    \end{equation}
    where \(\mathcal{K}_\omega\) is the set of points that appear exactly in \(\omega\) blocks. We now denote \(\kappa_\omega = |\mathcal{K}_\omega|\). Since each point should be present in at least one block (follows from the definition of multi-level covering design with \(m'=1\)), the collection \(\{\mathcal{K}_1,\mathcal{K}_2,\dots,\mathcal{K}_\Omega\}\) is clearly a partition of \([K]\). Therefore, we have
    \begin{equation}
    \label{eq:kappa1}\kappa_1+\kappa_2+\dots+\kappa_\Omega=K.
    \end{equation}
    Also, we have    
    \begin{equation}
    \label{eq:kappa2}\kappa_1+2\kappa_2+\dots+\Omega\kappa_\Omega\leq M\Omega
    \end{equation}
    since there are \(\Omega\) blocks\textcolor{black}{, each} of size at most \(M\).
    
    Now, consider two points \(k_1,k_2\in \mathcal{K}_1\). From the definition of multi-level covering design, for the set \(\{k_1,k_2\}\), we have
    \[|\left\{\omega:\{k_1,k_2\}\cap \mathcal{B}_\omega\neq \varnothing\right\}|\geq 2.\]
    Therefore, \(k_1\) and \(k_2\) must be in two distinct blocks. In general, no two points in \(\mathcal{K}_1\) can be present in the same block. Thus, we get
    \begin{equation}
    \label{eq:kappa3}
        \Omega\geq \kappa_1.
    \end{equation}
    We now find a lower bound on \(\kappa_1\) to get a lower bound on \(\Omega\). From \eqref{eq:kappa1} and \eqref{eq:kappa2}, we get
    \begin{align*}
    M\Omega&\geq \kappa_1+2\kappa_2+\dots+\Omega\kappa_\Omega\\
    &\geq \kappa_1+2(\kappa_2+\dots+\kappa_\Omega)\\
    & = \kappa_1+2(K-\kappa_1)\\
    & = 2K-\kappa_1.
    \end{align*}
    By rearranging the inequality, we get
    \begin{equation}
    \label{eq:kappa4}
        \kappa_1\geq 2K-M\Omega.
    \end{equation}
    Finally, by combining \eqref{eq:kappa3} and \eqref{eq:kappa4}, 
    we get the required lower bound on \(\Omega\) as
    \[\Omega\geq \frac{2K}{M+1}.\]
    Therefore, we can conclude that 
    \begin{equation*}
        R^*(K,L=2,M) = C(K,M,2) = \frac{2K}{M+1}
    \end{equation*}
    if \((M+1)|K\).
\end{proof}
\textcolor{black}{
We complete this part with a small corollary that extends some of the results to the case of the real and complex fields.  
\begin{cor}
    \label{cor:RealField}
    The achievability results of Theorem~\ref{thm1} and Theorem~\ref{thm2}, as well as the converse in Theorem~\ref{thmcovering}, remain valid even when our setting considers computations and communication over the real and complex fields. 
\end{cor} 
\begin{proof}
    In terms of achievable schemes, the proof is direct after observing that the nullspace-based argument in Lemma~\ref{lemma1} (see \eqref{eq:nullspace}) -- which relies on the notions of linear independence of vectors and the invertibility of the matrix \(\tilde{\mathbf{N}}_{\mathbf{D},\mathcal{M}}\) -- holds for demand matrices defined over any field. 
    Similarly, to conclude that the converse in Theorem~\ref{thmcovering} extends to the case of the field of real and complex numbers, it suffices to note that a) over the real or complex numbers it holds that for any \(L\) and \(K\), there exists a matrix \(\mathbf{D} \in \mathbb{R}^{L \times K}\) such that every \(L \times L\) submatrix of $\mathbf{D}$ has full rank \(L\), and b) by noting that the proof of Theorem~\ref{thmcovering} relies only on the matrix factorization in \eqref{eq:D=CA2}, which is valid even over the real and complex fields. 
\end{proof}    
}

\section{Tradeoff Between Rate and Number of Servers}
\label{RvsN}
The presented results until now considered the task of minimizing the communication cost $R$ without constraints on the number of servers. This current section places constraints on $N$, and thus explores the natural tradeoff between servers $N$ and rate $R$. 
     
\begin{thm}
    \label{thm6}
    For the $(K,L,M)$ distributed  linearly separable function computation problem, for every $\gamma\leq L$, the rate
    \begin{equation}
    \label{eq:thmtradeoff}
        R^{(\gamma)}(K,L,M) = \min\left\{K,L\left\lceil{\frac{K}{L+M-\gamma}}\right\rceil\right\}
    \end{equation}
    is achievable with
    \begin{equation}
    N^{(\gamma)}= \left\lceil{\frac{L}{\gamma}}\right\rceil\left\lceil{\frac{K}{L+M-\gamma}}\right\rceil
    \end{equation}
    servers. 
\end{thm}
\begin{proof}
    The achievability of \(R^{(\gamma)}(K,L,M)=K\) is trivial, and therefore we consider the scenario \(L\left\lceil{\frac{K}{L+M-\gamma}}\right\rceil<K\). Similar to the proof of Theorem~\ref{thm1}, we first consider the case \(K=L+M-\gamma\) for some \(\gamma\in [L]\), and show the achievability of \(R^{(\gamma)}(K,L,M)=L\) for any given \(\mathbf{D}\in \mathbb{F}_q^{L\times K}\) using \(N=\ceil[\Big]{\frac{L}{\gamma}}\) servers. Note that the condition \(L\left\lceil{\frac{K}{L+M-\gamma}}\right\rceil< K\) ensures that \(M-\gamma>0\), and consequently we have \(\gamma<M\) and \(K>L\). Furthermore, without loss of generality, we assume that \(\text{rank}_q(\mathbf{D})=L\). If this is not the case, we can replace some \(L-\text{rank}_q(\mathbf{D})\) linearly dependent rows of \(\mathbf{D}\) with a different set of \(L-\text{rank}_q(\mathbf{D})\) vectors, so that the resulting matrix has full row rank. If the user can decode the new functions, the removed functions can still be recovered as linear combinations of the original \(\text{rank}_q(\mathbf{D})\) functions that were retained. From the full-row-rank matrix \(\mathbf{D}\), we find an invertible submatrix \(\mathbf{D}_\mathcal{L}\) of size \(L\times L\), where \(\mathcal{L}= \{i_1,i_2,\dots,i_L\}\) denotes the indices of columns of \(\mathbf{D}\) chosen to form \(\mathbf{D}_\mathcal{L}\). We now give an explicit construction of the task assignment \(\mathcal{M}^\star\), where
    \begin{equation*}
        \mathcal{M}^\star = \{\mathcal{M}_1^\star,\mathcal{M}_2^\star,\dots,\mathcal{M}_{\ceil{\frac{L}{\gamma}}}^\star\}.
    \end{equation*}
For every \(\theta\in \left[\ceil[\Big]{\frac{L}{\gamma}}-1\right]\), we define
     \begin{equation}
         \label{eq:datasetassignl}
         \mathcal{M}_\theta^\star = \mathcal{M}_{\mathcal{L},\theta}^\star \bigcup \left(\left[K\right]\backslash \mathcal{L}\right)
     \end{equation}
     and
     \begin{equation}
         \label{eq:datasetassignLbyj}
         \mathcal{M}_{\ceil{\frac{L}{\gamma}}}^\star = \mathcal{M}_{\mathcal{L},\ceil{\frac{L}{\gamma}}}^\star \bigcup \left(\left[K\right]\backslash \mathcal{L}\right)
     \end{equation}
     where 
     \[\mathcal{M}_{\mathcal{L},\theta}^\star = \left\{i_{(\theta-1)\gamma+1},i_{(\theta-1)\gamma+2},\dots,i_{\theta \gamma}\right\}\] and \[\mathcal{M}_{\mathcal{L},\ceil{\frac{L}{\gamma}}}^\star=\left\{i_{(\ceil{\frac{L}{\gamma}}-1)\gamma+1},i_{(\ceil{\frac{L}{\gamma}}-1)\gamma+2},\dots,i_L\right\}.\]
Note that, \(|\mathcal{M}_\theta^\star|=\gamma+K-L=M\), for every \(\theta\in \left[\ceil{\frac{L}{\gamma}}-1\right]\), and \(|\mathcal{M}_{\ceil{\frac{L}{\gamma}}}^\star|=L-(\ceil{\frac{L}{\gamma}}-1)\gamma+K-L\leq M\). Notice that the collection of sets \(\{\mathcal{M}_{\mathcal{L},\theta}^\star:\theta\in \left[\ceil[\Big]{\frac{L}{\gamma}}\right]\}\) is a partition of \(\mathcal{L}\).

To now prove the achievability of \(R^{(\gamma)}(K,L,M)=L\), it is sufficient to show that \(\text{rank}_q(\mathbf{N}_{\mathbf{D},\mathcal{M}^\star})=L\), where \(\mathbf{N}_{\mathbf{D},\mathcal{M}^\star}\) is defined as per \eqref{eq:NullMat}. 
Suppose the rank of \(\mathbf{N}_{\mathbf{D},\mathcal{M}^\star}\) is less than \(L\), then there exists a row \(\ell'\) (where \(i_{\ell'}\in \mathcal{M}_{\mathcal{L},\theta}^\star\) for some \(\theta\in \left[\ceil[\Big]{\frac{L}{\gamma}}\right]\)) of \(\mathbf{N}_{\mathbf{D},\mathcal{M}^\star}\) 
     \begin{equation}
     \label{eq:contranew0}
         \mathbf{N}_{\mathbf{D},\mathcal{M}^\star}(\ell',:) = \sum_{\ell=1,i_\ell\not\in \mathcal{M}_{\mathcal{L},\theta}^\star}^L \alpha_\ell \mathbf{N}_{\mathbf{D},\mathcal{M}^\star}(\ell,:)
     \end{equation}
     where not all \(\alpha_\ell\)-s are zeros. Note that, for every \(\ell,\ell'\) such that \(i_{\ell},i_{\ell'}\in \mathcal{M}_{\mathcal{L},\theta}^\star\) for any \(\theta\), the rows \(\mathbf{N}_{\mathbf{D},\mathcal{M}^\star}(\ell,:)\) and \(\mathbf{N}_{\mathbf{D},\mathcal{M}^\star}(\ell',:)\) are linearly independent by definition. Now, on the one hand, we have
     \begin{equation}
         \label{eq:contranew1}
          \mathbf{N}_{\mathbf{D},\mathcal{M}^\star}(\ell',:)\mathbf{D}_\mathcal{L} = \sum_{\ell:i_\ell\in \mathcal{M}_{\mathcal{L},\theta}^\star}\beta_{\ell}\mathbf{e}_{\ell}^\intercal
     \end{equation}
     where \(\beta_{\ell}\) is non-zero for at least one \(\ell\) such that \(i_\ell\in \mathcal{M}_{\mathcal{L},\theta}^\star\). It follows from the fact that \(\mathbf{N}_{\mathbf{D},\mathcal{M}^\star}(\ell',:)\) is a vector in the left nullspace of \(\mathbf{D}_{({\mathcal{M}_{\theta}^\star})^c} = \mathbf{D}_{\mathcal{L}\backslash \mathcal{M}_{\mathcal{L},\theta}^\star}\). On the other hand, from \eqref{eq:contranew0}, we have
     \begin{align}
         \mathbf{N}_{\mathbf{D},\mathcal{M}^\star}(\ell',:)\mathbf{D}_\mathcal{L} &= \left(\sum_{\ell=1,i_\ell\not\in \mathcal{M}_{\mathcal{L},\theta}^\star}^L \alpha_\ell \mathbf{N}_{\mathbf{D},\mathcal{M}^\star}(\ell,:)\right)\mathbf{D}_\mathcal{L}\notag\\
         &= \sum_{\ell=1,i_\ell\not\in \mathcal{M}_{\mathcal{L},\theta}^\star}^L \alpha_\ell \bigg(\mathbf{N}_{\mathbf{D},\mathcal{M}^\star}(\ell,:)\mathbf{D}_\mathcal{L}\bigg)\notag\\
           &= \sum_{\ell=1,i_\ell\not\in \mathcal{M}_{\mathcal{L},\theta}^\star}^L \omega_\ell \mathbf{e}_\ell^\intercal\label{eq:contranew2}
     \end{align}
     where \(\omega_\ell=\alpha_\ell \beta_\ell \) are all not zeros. Clearly, \eqref{eq:contranew1} and \eqref{eq:contranew2} contradict. Therefore, \(\mathbf{N}_{\mathbf{D},\mathcal{M}^\star}\) cannot have rank less than \(L\). In other words, \(\text{rank}_q(\mathbf{N}_{\mathbf{D},\mathcal{M}^\star} )=L\), and therefore, for any given \(K,L\), and \(M\) with \(K=L+M-\gamma\), and for any \(\mathbf{D}\in \mathbb{F}_q^{L\times K}\), the rate 
     \begin{equation*}
        R^{(\gamma)}(K,L,M)=L
    \end{equation*}
    is achievable with \(N=\ceil[\Big]{\frac{L}{\gamma}}\).

For a general \(K\), we partition the demand matrix \(\mathbf{D}\in \mathbb{F}_q^{L\times K}\) into \(\ceil{K/K'}\) sub-demand matrices, where \(K'=L+M-\gamma\). That is 
\begin{equation*}
    \mathbf{D} = [\mathbf{D}_1,\mathbf{D}_2,\dots, \mathbf{D}_{\ceil{K/K'}}]
\end{equation*}
where \(\mathbf{D}_\nu\in \mathbb{F}_q^{L\times K'}\) for all \(\nu\in [\ceil{K/K'}-1]\), and \(\mathbf{D}_{\ceil{K/K'}}\in \mathbb{F}_q^{L\times (K-K'(\ceil{K/K'}-1))}\). 
Then for every \(\ell\in [L]\), as in \eqref{eq:aligned}, we get
\begin{equation}
\label{eq:Partitionnew}
    F_\ell(\mathcal{W}) = \sum_{\nu=1}^{\ceil{K/K'}} F_\ell^{(\nu)}(.)
\end{equation}
where \(F_\ell^{(\nu)}(.)\) are linearly separable functions of subsets of datasets (see \eqref{eq:Fellnu}). In addition, the supports of these functions (the domain of the functions)  for different values of \(\nu\) are disjoint. Therefore, using the scheme described for the case \(K'=L+M-\gamma\) in the first part of this proof, the user can retrieve \(F_\ell^{(\nu)}\), for every \(\ell\in [L]\), with a communication cost \(L\) and with a set of \(\ceil{L/\gamma}\) servers. This is true for every \(\nu \in [\ceil{K/K'}]\). Consequently, using \eqref{eq:Partitionnew}, the user can retrieve the demanded functions \(\{F_\ell(\mathcal{W}):\ell\in [L]\}\). Therefore, for every \(\gamma \in [L]\), the rate 
\begin{equation}
    R^{(\gamma)}(K,L,M) = L\ceil[\Big]{\frac{K}{K'}}=L\ceil[\Big]{\frac{K}{L+M-\gamma}}
\end{equation}
is achievable with \(N^{(\gamma)} = \ceil[\Big]{\frac{L}{\gamma}}\ceil[\Big]{\frac{K}{K'}}\) servers. This completes the proof of Theorem~\ref{thm6}.
\end{proof}

\begin{rem}
    \label{rem:RvsN}
   If \(\gamma = L\), the task assignment corresponds to a configuration in which no subfunction is computed on more than one server. This assignment strategy can be applied independently of the demand matrix and is often referred to as `disjoint placement' \cite{TLDK}. From \eqref{eq:thmtradeoff}, the rate corresponding to the disjoint task assignment is
\begin{equation*}
R_{\text{disjoint}}(K,L,M)=R^{(\gamma=L)}(K,L,M) = \min\left\{K, L \left\lceil \frac{K}{M} \right\rceil \right\}.   
\end{equation*}
The work \cite{KhE2}, which considers a multi-user setting with server-to-user connectivity where each user makes a single request, adopts a `disjoint support' assumption. This assumption differs from the disjoint task assignment. However, if the number of users is set equal to \(L\) and full connectivity is allowed, the model in \cite{KhE2} reduces to the setting considered in this work. Under these conditions, the `disjoint support' assumption in \cite{KhE2} coincides with the disjoint task assignment. It is also worth noting that
\begin{align*}
\frac{R_{\text{disjoint}}(K,L,M)}{R_1(K,L,M)}&=\frac{R^{(\gamma=L)}}{R^{(\gamma=1)}}\leq\frac{\min\left\{K, L \left\lceil \frac{K}{M} \right\rceil \right\}}{\min\left\{K, L \left\lceil \frac{K}{L+M-1} \right\rceil \right\}}\\
&\leq \frac{\min\left\{K, L \left\lceil \frac{K}{M} \right\rceil \right\}}{L \left\lceil \frac{K}{L+M-1} \right\rceil} \leq \begin{cases}
\frac{ L \left\lceil \frac{K}{M} \right\rceil }{\frac{LK}{L+M-1}} &\text{if} \ L<M\\
\frac{K}{\frac{LK}{L+M-1}} &\text{if \(L\geq M\)}
\end{cases}\\
&\leq \begin{cases}
\frac{  \frac{K}{M} +1 }{ \frac{K}{L+M-1} } &\text{if $L<M$}\\
\frac{L+M-1}{L} &\text{if \(L\geq M\)}
\end{cases}\quad \\ &\leq \begin{cases}
\frac{L+M-1}{M}+ \frac{L+M-1}{K}  &\text{if $L<M$}\\
\frac{L+M-1}{L} &\text{if \(L\geq M\)}
\end{cases}\\
&\leq \begin{cases}
2+ 1  &\text{if $L<M$}\\
2 &\text{if \(L\geq M\)}
\end{cases}\\&\leq 3.
\end{align*}
Directly from the above, we can also conclude that the entire achievable rate region characterized in Theorem~\ref{thm6} lies within a constant factor of 9 of our novel lower bound presented in Theorem~\ref{th:conv}.
\end{rem}

\begin{figure}[h]
\begin{center}   
\includegraphics[width=11cm]{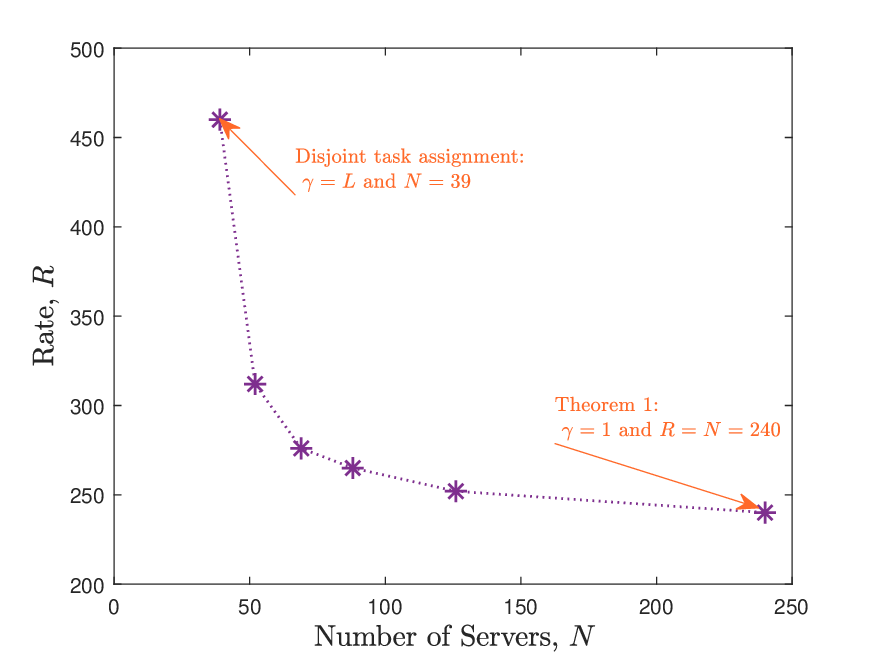}
\caption{Rate, \(R\) vs. number of servers, \(N\) for \(K=460,L=12,M=12\).}
\label{fig:RvsN1}
\end{center}
\end{figure}

\begin{rem}
    \label{rem:Grouping}
    Let \(\gamma_1<\gamma_2\leq L\) such that the conditions \(\gamma_1|L\) and \((L+M-\gamma_1)|K\), and \(\gamma_2|L\) and \((L+M-\gamma_2)|K\) hold. Then, from Theorem~\ref{thm6}, we get
\begin{equation*}
    \frac{R^{(\gamma_2)}}{R^{(\gamma_1)}}=\frac{L+M-\gamma_1}{L+M-\gamma_2}
\end{equation*}
and 
  \begin{equation*}
    \frac{N^{(\gamma_2)}}{N^{(\gamma_1)}}=\frac{\gamma_1(L+M-\gamma_1)}{\gamma_2(L+M-\gamma_2)}.
\end{equation*}  
Note that the penalty in rate is minimal, while we achieve a significant reduction in the required number of servers with no increase in per-server computation cost, especially when \(\gamma_1\) and \(\gamma_2\) are small.
\end{rem}
The tradeoff between rate and number of servers for the \((K=460,L=12,M=12)\) distributed linearly separable function computation problem is shown in Figure \ref{fig:RvsN1}.
Note that to achieve a rate \(R=240\), the system requires \(N=240\) servers. However, corresponding to \(\gamma=2\), a rate \(R=252\) is achievable with \(N=126\) servers. In other words, we could achieve almost the same rate with nearly half the number of servers compared to the previous case. Similarly, the \((N,R)\) pairs \((88,264), (69,276),(52,312),(39,460)\) are also achievable using Theorem~\ref{thm6}. Note that we have not included \((N,R)\) pairs corresponding to all possible values of \(\gamma\) between 1 and \(L\) in Figure \ref{fig:RvsN1}, since the ceiling functions in the expressions for \(R\) and \(N\) can make the curve non-decreasing for certain choices of \(\gamma\).

Suppose that the pairs \((N^{(\gamma_1)},R^{(\gamma_1)})\) and \((N^{(\gamma_2)},R^{(\gamma_2)})\) are achievable in the \(N-R\) curve. Then the pair \((\delta N^{(\gamma_1)}+(1-\delta)N^{(\gamma_2)},\delta R^{(\gamma_1)}+(1-\delta)R^{(\gamma_2)})\), \(\delta\in [0,1]\), is also achievable when certain divisibility conditions are satisfied. Let \(\gamma_1\) and \(\gamma_2\) be such that the conditions \(\gamma_1|L\) and \((L+M-\gamma_1)|K\), and \(\gamma_2|L\) and \((L+M-\gamma_2)|K\) hold. Then, using Theorem~\ref{thm6}, the pairs \((N^{(\gamma_1)},R^{(\gamma_1)})\) and \((N^{(\gamma_2)},R^{(\gamma_2)})\) are achievable, where \(N^{(\gamma_1)} = LK/(\gamma_1(L+M-\gamma_1))\), \(R^{(\gamma_1)} = LK/(L+M-\gamma_1)\), \(N^{(\gamma_2)} = LK/(\gamma_2(L+M-\gamma_2))\), and \(R^{(\gamma_2)} = LK/(L+M-\gamma_2)\). Further, assume that \(\delta \in [0,1]\) is such that the conditions \((L+M-\gamma_1)|\delta K\) and \((L+M-\gamma_2)|(1-\delta)K\) hold. We now show the achievability of the pair \((N,R)\), where \(N=\delta N^{(\gamma_1)}+(1-\delta)N^{(\gamma_2)}\) and \(R=\delta R^{(\gamma_1)}+(1-\delta)R^{(\gamma_2)}\). Let \(\mathbf{D}\in \mathbb{F}_q^{L\times K}\) be the demand matrix. Now, partition \(\mathbf{D}\) column-wise into two sub-demand matrices \(\mathbf{D}_1\in \mathbb{F}_q^{L\times \delta K}\) and \(\mathbf{D}_2\in \mathbb{F}_q^{L\times (1-\delta) K}\). Now, design task assignment and server transmissions using Theorem~\ref{thm6} on \(\mathbf{D}_1\) and \(\mathbf{D}_2\) with \(\gamma=\gamma_1\) and \(\gamma=\gamma_2\), respectively. Corresponding to \(\mathbf{D}_1\), the rate \(L(\delta K)/(L+M-\gamma_1)\) is achievable using \(L(\delta K)/(\gamma_1(L+M-\gamma_1))\) servers. Similarly, the rate \(L((1-\delta) K)/(L+M-\gamma_2)\) is achievable for the sub-demand matrix \(\mathbf{D}_2\) with \(L((1-\delta) K)/(\gamma_2(L+M-\gamma_2))\) servers. The user can decode the demanded functions as they are the sums of the corresponding sub-demand functions obtained from \(\mathbf{D}_1\) and \(\mathbf{D}_2\). Therefore, the rate 
\begin{align*}
    R &= \frac{L\delta K}{L+M-\gamma_1}+\frac{L(1-\delta)K}{L+M-\gamma_2}\\
    &=\delta R^{(\gamma_1)}+(1-\delta)R^{(\gamma_2)}
\end{align*}
is achievable using
\begin{align*}
    N &= \frac{L\delta K}{\gamma_1(L+M-\gamma_1)}+\frac{L(1-\delta)K}{\gamma_2(L+M-\gamma_2)}\\
    &=\delta N^{(\gamma_1)}+(1-\delta)N^{(\gamma_2)}
\end{align*}
servers.

\section{Conclusions}
In this work, we studied the problem of distributed computation of linearly separable functions and, under the assumption of linear encoding/decoding without subpacketization, characterized the optimal computation vs. communication-rate tradeoff to within a small constant factor across broad regimes. To establish the achievability results and tight converse bounds, we introduced novel design and analysis tools of broader relevance to coded distributed computing. On the achievability side, we introduced a nullspace-based design principle that simultaneously guides task assignment and server transmissions, ensuring exact decodability and attaining the global optimum when $K \leq L+M-1$. For $K>L+M-1$, we proposed two additional schemes; A first scheme based on column-wise partitioning that applies broadly and remains within a constant factor of optimal, and another, effective for $M \geq K/2$, that enjoys task assignments, at the servers, that are independent of the demand matrix.  
On the converse side, we reformulated the problem as a sparse matrix factorization to derive new information-theoretic lower bounds. These bounds not only certify the near-optimality of our schemes but also reveal a structural connection to combinatorial design theory, with the general covering number emerging as a fundamental limit. We further showed that communication efficiency degrades gracefully under server reduction, allowing comparable performance with fewer active servers.

Figures \ref{RvsM1} and \ref{RvsM2} illustrate the natural tradeoff between rate and computation cost $M$ in a $(K,L,M)$ distributed linearly separable function computing system. The rate is non-increasing in $M$, with extremes $R=K$ at $M=1$ and $R=L$ at $M=K$. The figures compare the achievable rates from Theorems~\ref{thm1} and \ref{thm2} with the lower bound of Theorem~\ref{th:conv} for various field sizes $q$, where discontinuities from the ceiling function in \eqref{eq:thm1} are most visible when $L$ is close to $K$ and the gap between the achievable rate and the converse diminishes as $q$ grows.

\begin{figure}[h!]
\begin{subfigure}{0.5\textwidth}
\centering
\includegraphics[width=1\linewidth]{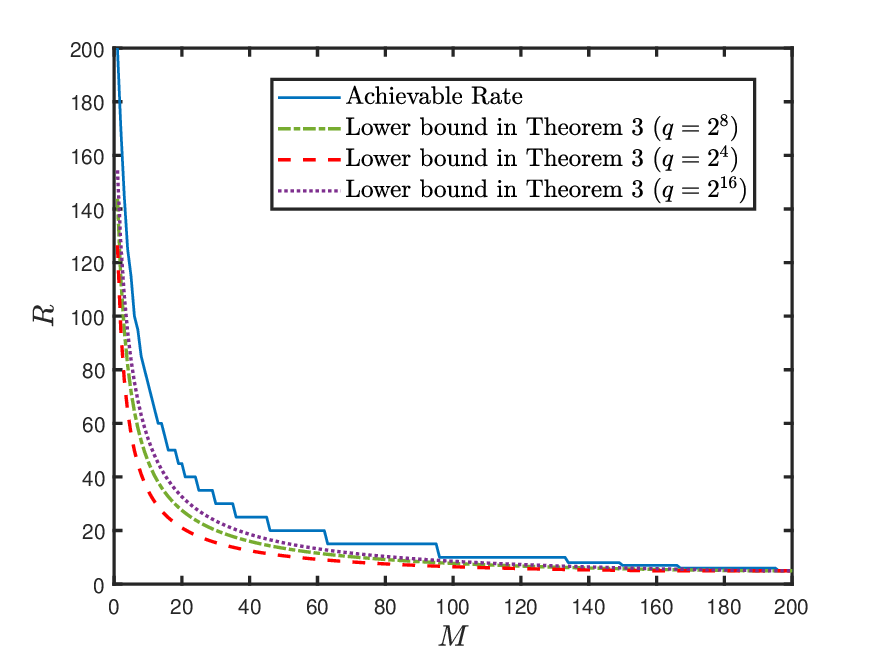} 
\caption{\((K=200,L=5,M)\) distributed linearly separable function computation problem.}
\label{RvsM1}
\end{subfigure}
\begin{subfigure}{0.5\textwidth}
\centering
\includegraphics[width=1\linewidth]{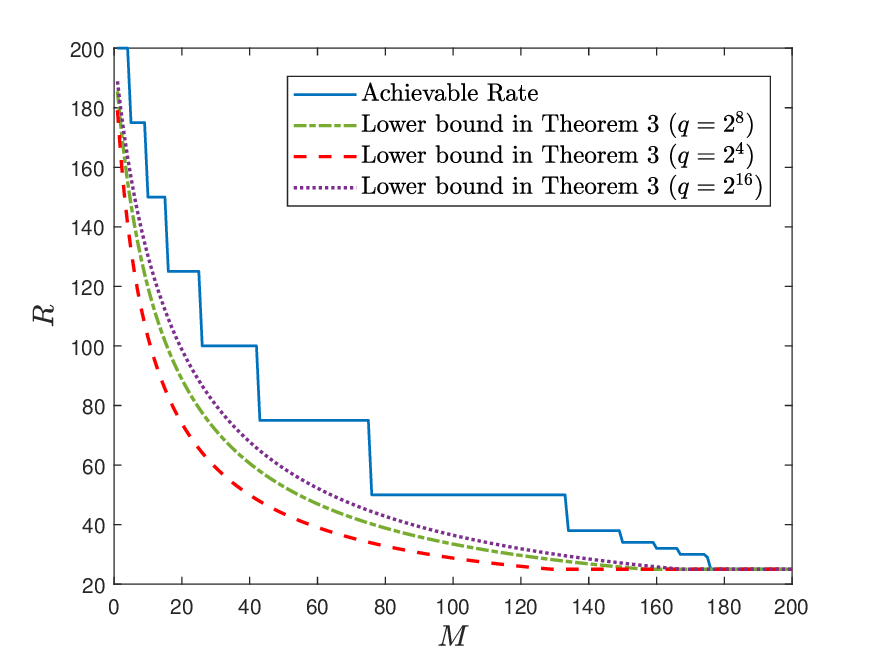}
\caption{\((K=200,L=25,M)\) distributed linearly separable function computation problem.}
\label{RvsM2}
\end{subfigure}
\caption{Rate \(R\) vs computation cost \(M\) for a fixed number of requests \(L\).}
\label{fig:RvsM}
\end{figure}

Similarly, Figures \ref{RvsL1} and \ref{RvsL2} illustrate the tradeoff between rate and the number of requested functions $L$ in $(K,L,M)$ distributed linearly separable function computation for $(K=200,M=5)$ and $(K=200,M=15)$. For fixed $K$ and $M$, the rate is generally non-decreasing in $L$, though the ceiling function in \eqref{eq:thm1} can introduce flat, non-monotonic segments. These flat regions appear in the achievable rate curves, while for large $L$ the behavior becomes linear since Scheme 1 achieves $R_1(K,L,M)=L$ whenever $L\geq K-M+1$.

\begin{figure}[h]
\begin{subfigure}{0.5\textwidth}
\centering
\includegraphics[width=\linewidth]{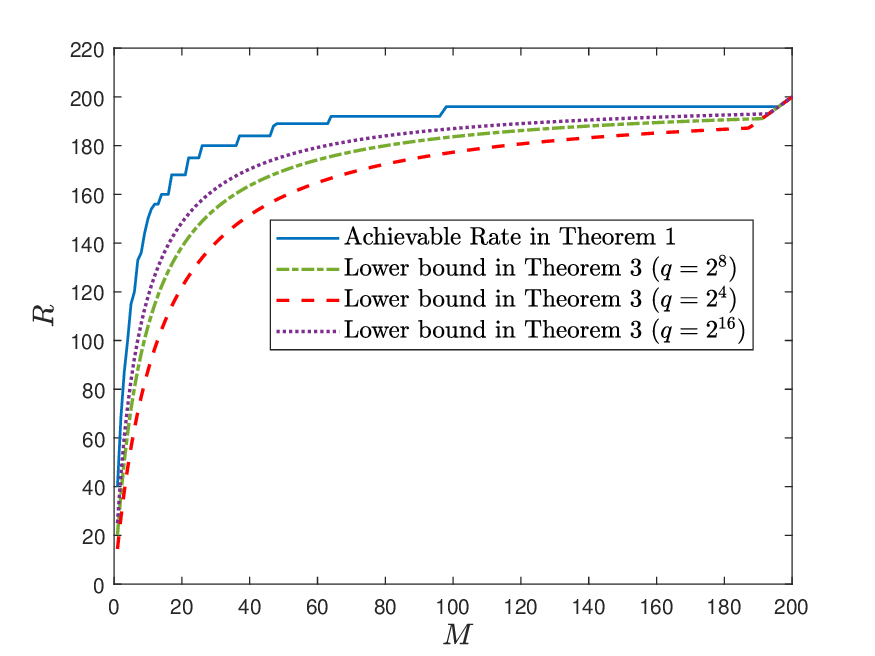} 
\caption{\((K=200,L,M=5)\) distributed linearly separable function computation problem.}
\label{RvsL1}
\end{subfigure}
\begin{subfigure}{0.5\textwidth}
\centering
\includegraphics[width=\linewidth]{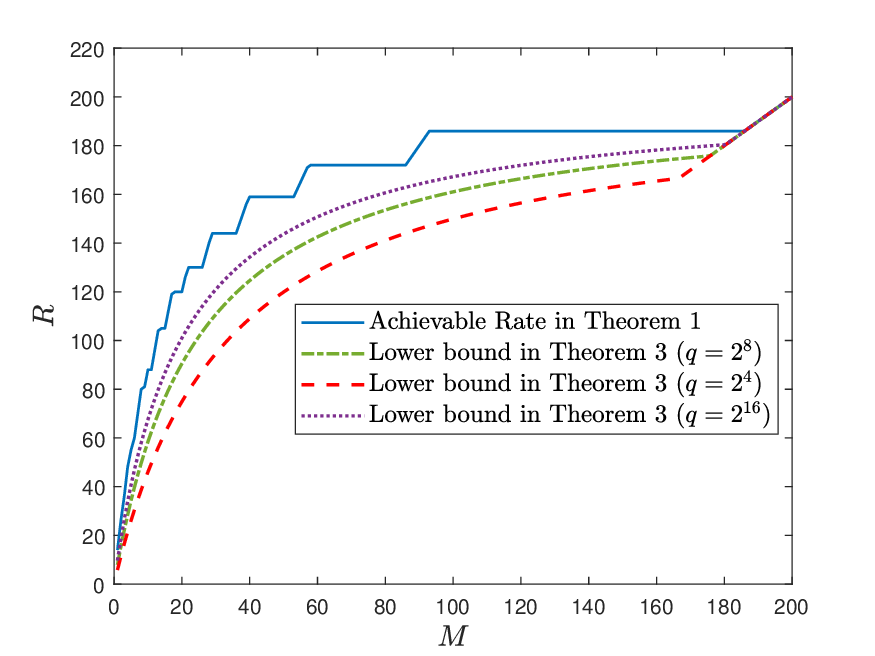}
\caption{\((K=200,L,M=15)\) distributed linearly separable function computation problem.}
\label{RvsL2}
\end{subfigure}
\caption{Rate \(R\) vs number of requests \(L\) for a fixed computation cost \(M\).}
\label{fig:RvsL}
\end{figure}

Several research problems remain open. A first direction is to address system-level constraints such as stragglers and heterogeneity across servers, where varying computational or communication capabilities may fundamentally reshape the optimal assignments and achievable tradeoffs. Another is to exploit data correlations, which, unlike the independent setting studied here, can be leveraged to reduce communication load and enable more efficient coding strategies.  Beyond these extensions, our framework can serve as a foundation for advancing AI-driven applications and modern distributed machine learning tasks. 

\bibliographystyle{IEEEtran}
\bibliography{Mybibliography}

\end{document}